\documentclass[a4paper,12pt]{amsart}

\usepackage[utf8]{inputenc}
\usepackage{amsmath}
\usepackage{graphicx}
\usepackage{natbib}

\usepackage{amssymb}
\usepackage{amsthm}
\usepackage{textcomp}
\usepackage{enumerate}
\usepackage{array}
\usepackage{verbatim}
\usepackage{hyperref}
\usepackage{booktabs}
\usepackage{algorithm}
\usepackage{algpseudocode}
\usepackage{stmaryrd} 
\usepackage{xcolor}
\hypersetup{
  colorlinks,
  linkcolor={red!50!black},
  citecolor={blue!50!black},
  urlcolor={blue!80!black}
}
\usepackage[lmargin=2.4cm,rmargin=2.25cm,bmargin=3cm,tmargin=2.5cm]{geometry}

\newcommand{\sref}[1]{\ref{#1}}
\newcommand{\appref}[1]{\ref{#1}}
\newcommand{\seqref}[1]{\eqref{#1}}
\newcommand{\mref}[1]{\ref{#1}}
\newcommand{\meqref}[1]{\eqref{#1}}

\newcommand{\defeq}{\mathrel{\mathop:}=}

\newcommand{\X}{\mathsf{X}}

\newcommand{\uarg}{\,\cdot\,}
\newcommand{\ud}{\mathrm{d}}

\renewcommand{\P}{\mathbb{P}}
\newcommand{\E}{\mathbb{E}}
\newcommand{\Var}{\mathrm{Var}}
\newcommand{\Cov}{\mathrm{Cov}}
\newcommand{\expm}{\mathrm{expm}}
\newcommand{\parexp}[1]{\exp{\Big(#1\Big)}}
\newcommand{\blocksize}{\mathrm{blocksize}}
\newcommand{\blocktime}{\mathrm{blocktime}}
\newcommand{\plu}{\mathrm{PLU}}
\newcommand{\iact}{\mathrm{IACT}}

\newcommand{\charfun}[1]{\mathbf{1}\left(#1\right)}

\newcommand{\given}{\,:\,}

\newcommand{\bigmid}{\;\big|\;}
\renewcommand{\vec}[1]{\boldsymbol{#1}}
\newcommand{\floor}[1]{\lfloor #1 \rfloor}

\newcommand{\Categ}{\mathrm{Categ}}

\newcommand{\ind}[1]{(#1)}

\newtheorem{theorem}{Theorem}

\newtheorem{lemma}[theorem]{Lemma}

\theoremstyle{definition}
\newtheorem{definition}[theorem]{Definition}
\newtheorem{assumption}[theorem]{Assumption}

\theoremstyle{remark}
\newtheorem{remark}[theorem]{Remark}

\begin{document}

\title{\bf Conditional particle filters with bridge backward sampling}

\author{Santeri Karppinen}
\address[SK]{Department of Mathematics and Statistics, University of Jyväskylä}

\author{Sumeetpal S.~Singh}
\address[SSS]{NIASRA, School of Mathematics and Applied Statistics,
University of Wollongong}

\author{Matti Vihola}
\address[MV]{Department of Mathematics and Statistics, University of Jyväskylä}

\def\spacingset#1{}

\begin{abstract}
Conditional particle filters (CPFs) with backward/ancestor sampling are powerful methods for sampling from the posterior distribution of the latent states of a dynamic model such as a hidden Markov model. However, the performance of these methods deteriorates with models involving weakly informative observations and/or slowly mixing dynamics. 
	Both of these complications arise when sampling finely time-discretised continuous-time path integral models, but can occur with hidden Markov models too. 
	Multinomial resampling, which is commonly employed with CPFs, resamples excessively for weakly informative observations and thereby introduces extra variance. 
	Furthermore, slowly mixing dynamics render the backward/ancestor sampling steps ineffective, leading to degeneracy issues. 
	We detail two conditional resampling strategies suitable for the weakly informative regime: the so-called `killing' resampling and the systematic resampling with mean partial order. 
	To avoid the degeneracy issues, we introduce a generalisation of the CPF with backward sampling that involves auxiliary `bridging' CPF steps that are parameterised by a blocking sequence. 
	We present practical tuning strategies for choosing an appropriate blocking. 
	Our experiments demonstrate that the CPF with a suitable resampling and the developed `bridge backward sampling' can lead to substantial efficiency gains in the weakly informative and slow mixing regime.
\end{abstract}


\maketitle

\section{Introduction} 

Sampling from the posterior distribution of the latent states of a dynamic model, such as a hidden Markov model (HMM), is a common statistical task. 
Such `smoothing problems' arise in many fields of science, such as ecology \citep{wood-ecology}, epidemiology \citep{rasmussen-epidemiology} and genetics \citep{mirauta-genetics}.
The particle Markov chain Monte Carlo (PMCMC) methods \citep{andrieu-doucet-holenstein} have proved invaluable for solving such problems. 
In particular, the conditional particle filter (CPF) with multinomial resampling and backward sampling (BS) \citep{andrieu-doucet-holenstein,whiteley-backwards-note}, or the probabilistically equivalent particle Gibbs with ancestor sampling algorithm \citep{lindsten-jordan-schon}, have been found to perform well with challenging HMMs and long data records \citep{fearnhead-kunsch,lee-singh-vihola}. 

However, the performance of CPF with BS (CPF-BS) deteriorate when the observations are weakly informative, that is, when the weights (or the potential functions in the Feynman--Kac (FK) distribution) are nearly constant \citep[cf.][]{chopin-singh-soto-vihola}. In such scenarios, the multinomial resampling steps introduce too much noise by eliminating particles excessively. Furthemore, when the dynamics of the model of interest are slowly mixing, the backward/ancestor sampling step has only a limited effect. We mitigate these issues by devising new resampling algorithms and an alternative for backward sampling for such weak potentials and slowly mixing scenarios. 

A time-discretisation of a continuous-time Feynman--Kac (FK) path integral model leads to an `extreme' weakly informative and slow mixing scenario \citep[cf.][]{chopin-singh-soto-vihola}.
Motivated by successes of CPF-BS in the usual discrete time settings, we were interested to seek for a CPF-BS analogue which is stable with respect to refined time-discretisations. 
It is relatively easy to see that the direct application of BS degenerates under such refined discretisations, except for limited cases such as when the driving Markov process admits jumps, such as considered in \citep{miasojedow-niemiro}, or in a univariate case where the trajectories can cross with positive probability.

To address the inherent inefficiency of multinomial resampling, we draw inspiration from the recent works of \cite{arnaudon-delmoral} and \cite{chopin-singh-soto-vihola} where other types of resampling are shown to be more effective for weakly informative potentials. \cite{arnaudon-delmoral}  propose a continuous-time version of the CPF with `killing' resampling, however, this is an idealised algorithm in the sense that practical diffusion models need to be time-discretised. The work of \cite{chopin-singh-soto-vihola} studies the stability of resampling for particle filters (and not the CPF) as the time discretisation is refined. 
There, a new systematic resampling method is proposed, which incorporates a `mean partition' step, which has not been developed for the CPF yet. The mean partition step does further decrease superfluous resampling \citep{chopin-singh-soto-vihola},
and empirical results suggest improved performance.

The main contributions of this paper are as follows.
\begin{itemize}
  \item We detail two new conditional resampling algorithms: the `killing' and systematic resampling with mean partition (Section \ref{sec:conditional-resamplings}). These are conditional versions of resampling algorithms, which were shown to be stable in the weak potentials setting (in the continuous-time limit) \citep{chopin-singh-soto-vihola}. 
    We also detail a generic sufficient condition for conditional resamplings (Assumption \ref{a:conditionals}), which guarantees validity of the CPF (Theorem \ref{thm:cpf-valid}), and complements the result of \cite{chopin-singh}.
  \item Our main contribution is a new CPF with bridge backward sampling (CPF-BBS) (Section \ref{sec:cpf-bbs}), which may be regarded as a generalisation of BS to an arbitrary `blocking sequence,' and which can avoid the degeneracy problem of CPF-BS with refined discretisations.
  \item The performance of the CPF-BBS relies on an appropriately chosen blocking sequence, which depends on the model at hand. Therefore, a significant portion of our work focuses on finding practical, computationally inexpensive and robust tuning criteria for choosing such a sequence (Section \ref{sec:blocking-seq-selection}). 
    We introduce a method for blocking sequence selection that requires a small number of independent runs of the standard particle filter for the model of interest. 
\end{itemize}

Our developments related to blocking sequence selection can be of independent interest, and potentially useful with other methods based on blocking, such as the blocked particle Gibbs \citep{singh-lindsten-moulines}.
Systematic resampling in the context of CPF has been proposed before \citep{chopin-singh} but not the more efficient mean partition version. Furthermore, \cite{chopin-singh} do not discuss or demonstrate efficiency in the context of weak potentials, which is our primary motivation.

The CPF-BBS is a general method, but requires evaluation of and simulation from the conditional distributions of (multiple steps of the) proposal distributions. In practice, this typically means that the proposals are linear-Gaussian, arising for instance from a linear stochastic differential equation (SDE). 
The latter can occur in single molecule studies \citep{dAvigneau-SMM}, and one of our numerical examples demonstrates how animal movement modelling based on telemetry data \citep{johnson-ctcrw} can be combined with a path integral model to account for habitat preferences instead of so-called step-selection formulation \citep[cf.][]{hooten-johnson-mcclintock-morales,thurfjell-ciuti-boyce}. 
Linear-Gaussian state dynamics are common with structural time series models \citep{durbin-koopman} too, and smoothing distribution approximations can lead to weak potentials \citep[cf.][]{vihola-helske-franks}.

The CPF-BBS features `bridging' CPF steps, which resemble the intermediate block importance sampling suggested in \citep{lindsten-bunch-singh-schon}, and the MCMC rejuvenation considered in \citep{lindsten-bunch-singh-schon,carter-mendes-kohn}; there are similarities also with the bridging particle filter suggested in \citep{delmoral-murray}; see also \citep{mider-schauer-vandermeulen}. We believe that our approach is more efficient than direct importance bridging, and because our approach can be intuitively related to a continuous-time analogue (through \citep{chopin-singh-soto-vihola}), it is expected to behave well with respect to refinement of time-discretisation, unlike the MCMC bridging.

Our experiments (Section \ref{sec:experiments}) demonstrate how the developed resamplings outperform standard multinomial resampling in the weak potentials and slow mixing scenario, and we establish
empirically an order between their performance, which follows a similar pattern as the results of \cite{chopin-singh-soto-vihola} for the standard particle filter applied to FK path integral models.
Empirical results using the CPF-BBS show a significant improvement over CPF-BS, and reveal how the method is stable with respect to refined discretisation.
Finally, our tuning algorithm appears to deliver blocking sequences that facilitate efficient inference with little additional specification from the user.

\section{Preliminaries and notation} \label{sec:notation} 

We aim at inference of a smoothing distribution, typically arising from a hidden Markov model (HMM), which consists of a latent Markov state $X_{1:T}$ taking values in $\X$, with initial (prior) density $m_1(x_1)$ and transition probability densities $(m_k(x_k \mid x_{k-1}))_{2 \leq k \leq T}$, as well as conditionally independent observations $Y_{1:T}$ with observation densities $(g_k(y_k \mid x_k))_{1 \leq k \leq T}$.
The smoothing distribution is the posterior density of the latent states $X_{1:T}$:
\begin{equation}
  \pi(x_{1:T}) = \frac{\gamma(x_{1:T})}{\mathcal{Z}} ,\quad
  \text{where}\quad \gamma(x_{1:T})\defeq m_1(x_1) g_1(y_1 \mid x_1)
	\prod_{k=2}^T m_k(x_k\mid x_{k-1}) g_k(y_k \mid x_k), 
  \label{eq:target-hmm}
\end{equation}
with normalising constant $\mathcal{Z}\defeq \int \gamma(x_{1:T})\ud x_{1:T}\in(0,\infty)$.

We approach the inference of \eqref{eq:target-hmm} via an alternative representation often referred to as a Feynman--Kac (FK) distribution \citep{del-moral}.
In our context, the FK distribution consists of the `components' $(M_{1:T}, G_{1:T})$ (described below) that are chosen such that
\begin{equation}
	\label{eq:target}
	\gamma(x_{1:T}) \propto M_1(x_1) G_1(x_1) \prod_{k=2}^T M_k(x_k\mid x_{k-1}) G_k(x_{k-1:k}), 
\end{equation}
where $\gamma(x_{1:T})$ is defined as in \eqref{eq:target-hmm}.
The components $M_1(x_1)$ and $M_k(x_k\mid  x_{k-1})$ define, respectively, an alternative initial
distribution and `proposal' transition densities of a Markov chain on $\X$, and the components $G_1:\X\to[0,\infty)$ and $G_k:\X^2\to[0,\infty)$ for $k \geq 2$ are called `potential' or `weight' functions.

The FK distribution \eqref{eq:target} is a generalisation of the HMM smoothing problem \eqref{eq:target-hmm}. Indeed, taking $M_k \equiv m_k$ and $G_1(x_1) = g_1(y_1^* \mid x_1)$ and $G_k(x_{k-1:k}) \equiv g_k(y_k^*\mid x_k)$ satisfies \eqref{eq:target}. 
However, it is often beneficial to choose another `proposal' family $M_k\not\equiv m_k$, in which case the `weights' in \eqref{eq:target} may be taken as $G_1(x_1) \defeq g_1(y_1^*\mid x_1) m_1(x_1)/M_1(x_1)$ and $G_k(x_{k-1},x_k) \defeq g_k(y_k^*\mid x_k) m_k(x_k\mid x_{k-1})/M_k(x_{k}\mid x_{k-1})$ for $k\ge 2$ to satisfy \eqref{eq:target}.
The FK distribution \eqref{eq:target} arises also in contexts beyond the HMM scenario, such as with time-discretised path-integral models, as discussed in Section \ref{sec:fk-path-integral}.

Above and hereafter, `$\ud x$' stands for a $\sigma$-finite dominating measure on a general state space $\X$, integers are equipped with the counting measure, and product spaces are equipped with products of the dominating measures. We use the shorthand notation 
for sequences: for $\{x_i\}_i$, $\{y^j\}_j$ and $\{z_i^j\}_{i,j}$ we write $x_{a:b} = (x_a,\ldots,x_b)$, $y^{\ind{a:b}} = (y^{\ind{a}},\ldots,y^{\ind{b}})$ and $z_{a:b}^{\ind{j_{a:b}}} = (z_a^{\ind{j_a}},\ldots,z_b^{\ind{j_b}})$.
 We also denote $[N]\defeq\{1,\ldots,N\}$, and use $\mathcal{U}(a, b)$ and $\mathcal{U}([N])$ to denote the uniform distribution in the interval (a,b) and on integers $[N]$, respectively. 
Finally, test and potential functions are implicitly assumed measurable.
  
 \section{The conditional particle filter} \label{sec:cpf} 

The particle filter is a sequential Monte Carlo algorithm, which includes sampling from Markov dynamics $M_k$, and \emph{resampling} proportional to  weights arising from $G_k$. The resampling operation $r(a^{\ind{1:N}}\mid g^{\ind{1:N}})$ defines a probability distribution on $[N]^N$ which depends on  non-negative `unnormalised weights' $g^{\ind{1:N}}$. That is, if $A^{\ind{1:N}}\sim r(\uarg\mid g^{\ind{1:N}})$ in $[N]$, then $\P(A^{\ind{1:N}} = a^{\ind{1:N}}) = r(a^{\ind{1:N}}\mid g^{\ind{1:N}})$. We will only consider \emph{unbiased resamplings} \citep{crisan-discrete} $r$, which means that 
for all $j\in[N]$:
\begin{equation}
\label{eq:res-unbiased} 
  \Bigg( \sum_{i=1}^N g^{\ind{i}} \Bigg) \E_{r(\uarg\mid g^{\ind{1:N})})}
\bigg[ \frac{1}{N}\sum_{i=1}^N \charfun{A^{\ind{i}} = j}\bigg] 
= g^{\ind{j}}.
\end{equation}

Algorithm \ref{alg:pf} describes the particle filter targeting the FK distribution \eqref{eq:target} using $N$ particles, 
and an unbiased resampling $r(\uarg\mid g^{\ind{1:N}})$.
The boldface notation $\mathbf{X}_1^{(i)} = X_1^{(i)}$ and $\mathbf{X}_{k+1}^{(i)} = ({X}_{k}^{(A_k^{(i)})},X_{k+1}^{(i)})$ 
stands for the latest particles augmented with their ancestors, generated during the algorithm. 
\begin{algorithm}
	\spacingset{1.0}
  \caption{\small \textsc{PF}$(r, M_{1:T}, G_{1:T}, N)$}
    \label{alg:pf} 
\begin{algorithmic}[1]
  \small
  \State Draw $X_1^{(i)} \sim M_1(\cdot)$ for $i \in [N]$.
  \State Set $\mathbf{X}_1^{(i)} = X_1^{(i)}$ for $i \in [N]$
  \For{$k = 1, \ldots, T-1$}

  \State Set $W_k^{(i)} = G_k(\mathbf{X}_k^{(i)})$ for $i \in [N]$
    \State Draw $A_{k}^{(1:N)} \sim r(\cdot \mid W_{k}^{(1:N)})$
    
    \State Draw $X_{k + 1}^{(i)} \sim M_{k + 1}(\cdot \mid X_{k}^{(A_{k}^{(i)})})$ for $i \in [N]$
    
    \State Set $\mathbf{X}_{k + 1}^{(i)} = (X_{k}^{(A_{k}^{(i)})}, X_{k + 1}^{(i)})$ for $i \in [N]$ 
  \EndFor
    \State Set $W_{T}^{(i)} = G_{T}(\mathbf{X}_{T}^{(i)})$ for $i \in [N]$
  \State \textbf{output} $(X_{1:T}^{(1:N)}, A_{1:T-1}^{(1:N)}, W_{1:T}^{(1:N)})$ 
  
\end{algorithmic} 
\end{algorithm}
In what follows, we use underline to denote `all particles' at one time instant, so for instance $\underline{X}_k = X_k^{(1:N)}$.

The conditional particle filter (CPF) introduced in \citep{andrieu-doucet-holenstein} is similar to the particle filter, but its motivation is different: it implements a $\pi$-invariant Markov transition kernel from an input (so called `reference' path) $X_{1:T}^* \in \X^{T}$ to a newly chosen path $\tilde{X}_{1:T}^* \in \X^{T}$.
The original scheme of \cite{andrieu-doucet-holenstein} assumed multinomial sampling with ancestor tracing.
Algorithm \ref{alg:cpf} presents a generic version of the CPF with $N$
particles and with ancestor tracing, using a \emph{generic conditional
resampling} $r^{(p,n)}$. 
\begin{algorithm}
	\spacingset{1.0}
  \caption{\small CPF-AT$(r^{(p, n)}, X_{1:T}^*,B_{1:T}; M_{1:T}, G_{1:T}, N)$}
    \label{alg:cpf} 
\begin{algorithmic}[1]
  \small
\State
  $(\underline{X}_{1:T}, \underline{A}_{1:T-1}, \tilde{B}_T) \gets \mathrm{CPF}(r^{(p, n)}, X_{1:T}^*,B_{1:T}; M_{1:T}, G_{1:T}, N)$
\State $\tilde{B}_{1:T-1} \gets \textsc{AncestorTrace}(\underline{A}_{1:T-1},
\tilde{B}_T)$
\State \textbf{output} $(\tilde{X}_{1:T}^*, \tilde{B}_{1:T})$ where $\tilde{X}_{1:T}^* =
X_{1:T}^{(\tilde{B}_{1:T})}$
\end{algorithmic}
\end{algorithm}
\begin{algorithm}
	\spacingset{1.0}
  \caption{\small CPF$(r^{(p, n)}, X_{1:T}^*,B_{1:T}; M_{1:T}, G_{1:T}, N)$}
    \label{alg:cpf-core} 
\begin{algorithmic}[1]
  \small
\State
  Draw $X_1^{(-B_1)} \sim M_1(\uarg)$ and set $X_1^{(B_1)}\gets X_1^*$ and $\mathbf{X}_1^{(i)} = X_1^{(i)}$ for $i \in [N]$
\label{item:cpf-initial-samples}
\For{$k=1,\ldots,T-1$}
\State $A_{k}^{(1:N)} \gets r^{(B_k,B_{k+1})}\big(\uarg \mid
  G_{k}(\vec{X}_{k}^{(1:N)})\big)$ \label{line:conditional-res}
\State Draw $X_{k+1}^{(i)} \sim M_{k+1}(\uarg \mid
            X_{k}^{(A_{k}^{(i)})})$ for
            $i\neq B_{k+1}$ and
            set $X_{k+1}^{(B_{k+1})} = X_{k+1}^*$
\State Set $\mathbf{X}_{k+1}^{(i)} = (X_k^{(A_{k}^{(i)})}, X_{k+1}^{(i)})$ for $i \in [N]$
\EndFor
\State Draw $\tilde{B}_T\sim\Categ\big(G_T(\vec{X}_T^{(1:N)})\big)$
\label{item:cfp-drawfinal}
\State \textbf{output} $(\underline{X}_{1:T}, \underline{A}_{1:T-1},
\tilde{B}_T)$
\end{algorithmic}
\end{algorithm}
\begin{algorithm}
	\spacingset{1.0}
\caption{\small \textsc{AncestorTrace}$(\underline{a}_{\ell:u-1}, b_u)$}
\label{alg:at} 
\begin{algorithmic}
    \small
    \State \textbf{for} $v = u-1,u-2,\ldots, \ell$ \textbf{do} $b_v  \gets a_v^{(b_{v+1})}$
    \State \textbf{output} $b_{\ell:u-1}$
\end{algorithmic}
\end{algorithm}
The conditional resampling scheme draws the ancestor indices (on line \ref{line:conditional-res} of Algorithm \ref{alg:cpf-core}) conditional on the ancestor of the reference.
This makes it possible to write Algorithm \ref{alg:cpf} such that the reference trajectory can be located at arbitrary indices $B_{1:T}$, unlike earlier formulations, which assume reference at index 1 \citep[e.g.][]{chopin-singh}.
The arbitrary reference indices turn out to be convenient for us, when we introduce the bridge backward sampling CPF in Section \ref{sec:cpf-bbs}.
Definition \ref{def:valid-conditional-resampling} gives a sufficient condition that $r^{(p, n)}$ is a valid conditional resampling for use with Algorithm \ref{alg:cpf}.

\begin{definition}
    \label{def:valid-conditional-resampling} 
The conditional resampling scheme $r^{(p,n)}(\uarg \mid g^{(1:N)})$
is valid, if there exists an unconditional resampling scheme $r(\uarg\mid g^{(1:N)})$, such that
for all $g^{(1:N)}\ge 0$ such that $\sum_{\ell=1}^N g^{(\ell)}>0$,
and all $p,n\in\{1{:}N\}$,
\begin{enumerate}[(i)]
\item $\P_{r^{(p, n)}(\uarg \mid g^{(1:N)})}(A^{(n)} = p)=1$,
\item $r^{(p,n)}(a^{(1:N)}\mid g^{(1:N)}) = \P_{r(\uarg\mid
    g^{(1:N)})}(A^{(-n)}=a^{(-n)}\mid A^{(n)} = p)$,
\item $\P_{r(\uarg\mid g^{(1:N)})}(A^{(n)}=p) =
  \frac{g^{(p)}}{\sum_{i=1}^N g^{(i)}}$. \label{defcond:symmetricity}
\end{enumerate}
\end{definition}

\begin{theorem}
    \label{thm:cpf-valid} 
  Algorithm \ref{alg:cpf} with a valid conditional resampling $r^{(p, n)}$
defines
a Markov update $(X_{1:T}^*, B_{1:T}) \to
(\tilde{X}_{1:T}^*, \tilde{B}_{1:T})$ that is reversible
with respect to $\pi\times \mathcal{U}([N]^T)$.
\end{theorem} 
Theorem \ref{thm:cpf-valid}, whose proof is given in Appendix
\appref{app:cpf}, complements the result of \cite{chopin-singh} by accomodating our version of the CPF, where the reference is placed at arbitrary position, and allows for the resamplings which we discuss in Section \ref{sec:conditional-resamplings}.

Ancestor tracing version of the CPF requires typically $N = O(T)$ in order to remain efficient \citep{andrieu-lee-vihola,lindsten-douc-moulines}. If the potentials $G_k$ are nearly uniform, this issue can be mitigated to some extent by an efficient resampling. For instance, when the model is a time discretisation of a continuous-time model (cf.~Section \ref{sec:fk-path-integral}), $N$ does not need to increase with respect to the number of time steps, but in general, $N$ still has to increase with respect to the time horizon. 

\cite{whiteley-backwards-note} suggested, in a discussion note to \citep{andrieu-doucet-holenstein}, that if transition densities $M_k$ can be calculated, backward sampling may be used in CPF (with multinomial resampling), instead of ancestor tracing. Later, a probabistically equivalent variant of BS called `ancestor sampling' (AS) \citep{lindsten-jordan-schon}, was introduced. The Markov kernels of CPF with BS/AS are reversible with respect to $\pi$, and guaranteed to outperform AT in the asymptotic variance sense \citep{chopin-singh}. The improvement has been found substantial in many empirical studies, and the BS/AS variants have been found stable with constant $N$ and very large $T$; see \citep{lee-singh-vihola} for a theoretical result which consolidates these findings. We present a generalisation of BS in Section \ref{sec:cpf-bbs}, which remains efficient with slowly mixing $M_k$, and in particular, with refined time discretistaions.

\section{New conditional resampling algorithms} \label{sec:conditional-resamplings} 

The simplest unbiased resampling, that is, satisfying \eqref{eq:res-unbiased} is \emph{multinomial resampling}, where $A^{(k)}$ are drawn independently from the categorical distribution $\Categ(w^{\ind{1:N}})$ with normalised weights $w^{\ind{j}} =\frac{g^{\ind{j}}}{\sum_{i=1}^N g^{\ind{i}}}$. 
However, in the context of this work, multinomial resampling is wasteful, and we focus instead on conditional versions of two resampling algorithms, that were recently shown to be stable in refined discretisations \citep[][]{chopin-singh-soto-vihola}. We briefly discuss these resampling algorithms below, and then describe their conditional variants.

The first resampling is the `killing' resampling, defined as follows \citep[cf.][]{del-moral-2013}:
\begin{align}
    r_{\mathrm{kill}}(a^{(1:N)} \mid g^{(1:N)}) 
    &\defeq \prod_{i=1}^N \bigg[\charfun{a^{(i)} = i}
    \frac{g^{(i)}}{g^*}  
    + \Big(1- \frac{g^{(i)}}{g^*}\Big) 
    \sum_{j=1}^N \charfun{a^{(i)} = j} 
    \frac{g^{(j)}}{\sum_{\ell=1}^N g^{(\ell)}}
    \bigg],
    \label{eq:killing-resampling}
\end{align}
where $g^* = \max_{i\in\{1:N\}} g^{(i)}$ (and in case $g^*=0$,
$r_{\mathrm{kill}}$ may be defined arbitrarily). The killing resampling is valid
also with any other choice of $g^*$ as long as $g^{(j)}\le g^*$, but we consider the above one minimising the resampling rate, which was also used in \citep{murray-lee-jacob}. Like multinomial resampling, the components of the random vector $A^{(1:N)}$ are independent but not identically distributed.

The second resampling is `systematic resampling with mean partition' (Definition \ref{def:sys-res-mean-partition}), which is a variant of `standard' systematic resampling (Definition \ref{def:sys-res}) where the weights $w^{1:N}$ are processed in a particular `mean partition' order (Definition \ref{def:mean-partition}).
\begin{definition}
  \label{def:sys-res}
  (Systematic resampling). Input normalised weights $w^{1:N}$.
  Simulate a single $\tilde{U} \sim \mathcal{U}(0, 1)$, set 
  $
    \check{U}^{i} \defeq (i - 1 + \tilde{U})/N
  $
  and define the resampling indices as $A^{i} \defeq F^{-1}(\check{U}^i)$ for $i \in [N]$.
  Here, the generalised inverse $F^{-1}(u)$ is defined for $u \in (0, 1)$ as the unique index 
  $i \in [N]$ such that $F(i - 1) < u \leq F(i)$, with $F(i) \defeq \sum_{j = 1}^{i} w^j$.
\end{definition}

\begin{definition}
  \label{def:mean-partition}
  (Mean partition order) Suppose that $u^{1:N} \in \mathbb{R}^N$.
  A permutation $\varpi: [N] \rightarrow [N]$ is a \textit{mean partition} order for $u^{1:N}$, if the re-indexed vector $u_{\varpi}^{i} \defeq u^{\varpi(i)}$ satisfies $u_{\varpi}^{1}, \ldots, u_{\varpi}^{m} \leq \bar{u}$ and $u_{\varpi}^{m + 1}, \ldots, u_{\varpi}^{N} > \bar{u}$ for some $m \in [N]$, with $\bar{u}$ denoting the mean of the vector $u$.
\end{definition}

\begin{definition}
  \label{def:sys-res-mean-partition}
  (Systematic resampling with mean partition).
  Let $F_{\varpi}^{-1}$ denote the generalised inverse distribution function corresponding to the 
  re-indexed weights $w_{\varpi}^{1:N}$, where $\varpi$ is a mean partition order as in Definition \ref{def:mean-partition}.
  Set $A^{i} \defeq \varpi(F_{\varpi}^{-1}(\check{U}^i)),$
  where $\check{U}^{1:N}$ are defined as in Definition \ref{def:sys-res}.
\end{definition}
Systematic resampling with mean partition was introduced in \citep{chopin-singh-soto-vihola}, where it was shown to have the smallest resampling rate among a number of other algorithms, including killing resampling. 
The empirical evidence in \citep{chopin-singh-soto-vihola} further suggests that
the mean partition order variant improves performance over the `standard' systematic resampling.
The work of \cite{chopin-singh-soto-vihola} featured also another resampling algorithm, the Srinivasan sampling process \citep{gerber-chopin-whiteley}, which was on a par with systematic.
However, it seems difficult to devise an efficient implementation of the conditional version of the latter.

Algorithms \ref{alg:killing} and \ref{alg:systematic} describe the conditional variants of killing resampling and systematic resampling with mean partition, respectively. 
These algorithms are valid conditional resamplings, as indicated by Lemma \ref{lem:conditional-killing}, whose proof is given in Appendix \appref{app:cpf}.
Both algorithms make use of a cyclic shift of $1{:}N$ by $s$, which is denoted by $\sigma_{s}^{N}$, that is, $\sigma_{s}^{N}(i) := 1 + (i + s - 1) \mod N$.
The conditional systematic resampling with mean partition is an extension of the conditional systematic resampling of \cite{chopin-singh}.
The mean partition may be found in $O(N)$ time, and our implementation is based on Hoare's scheme \cite{hoare-quicksort}; see Algorithm \sref{alg:mean-partition-order} in Appendix \appref{sec:misc-algs}. 
\begin{algorithm}
	\spacingset{1.0}
  \caption{\small Conditional killing resampling $\rho_{\mathrm{kill}}^{(i,k)}(\uarg\mid
      g^{(1:N)})$.}
    \label{alg:killing} 
\begin{algorithmic}[1]
  \small
  \State Draw $\bar{A}^{(1:N)} \sim \rho_{\mathrm{kill}}(\uarg\mid g^{(1:N)})$
\State Draw $J\in[N]$ such that
$
\P(J=j) = h(j\mid i) \defeq \begin{cases}
\frac{1}{N}\Big( 1 + \frac{\sum_{\ell\neq i} g^{(\ell)}}{g^*}\Big),
  & j = i \\
\frac{1}{N}\Big( 1 - \frac{g^{(j)}}{g^*}\Big), & j \neq i
\end{cases}
$
\State Set $S \defeq \sigma_{-k}^N(J)$ and 
$\bar{A}^{(J)} \gets i$
\State \textbf{output} $A^{(1:N)}$ where $A^{(j)} = \bar{A}^{(\sigma_S^N(j) )}$
\end{algorithmic}
\end{algorithm}

\begin{algorithm}
	\spacingset{1.0}
  \caption{\small Conditional systematic resampling with mean partition $\rho_{\mathrm{syst}}^{(i,k)}(\uarg\mid
  g^{(1:N)})$.}
    \label{alg:systematic} 
\begin{algorithmic}[1]
  \small
	\State Define $W^j := \dfrac{g^{(j)}}{\sum_{l=1}^{N} g^{(l)}}$ for $j \in [N]$
	\State Set $r = NW^i - \floor{NW^i}$ and $p = \dfrac{r (\floor{N W^i} + 1)}{NW^i}$
	\State With probability $p$, draw $\bar{U} \sim \mathcal{U}(0, r)$ and set $N^i = \floor{NW^i} + 1$; otherwise draw $\bar{U} \sim \mathcal{U}(r, 1)$ and set $N^i = \floor{NW^i}$
  \State Set $\varpi \gets \textsc{MeanPartitionOrder}(W^{1:N})$ \algorithmiccomment{Algorithm \sref{alg:mean-partition-order} in Appendix \appref{sec:misc-algs}}
	\State Set $s = \varpi^{-1}(i)$ and define $\tilde{\varpi}$ such that $\tilde{\varpi}(j) = \varpi(\sigma_{s-1}^N(j))$ for $j \in [N]$
	\State Set $\bar{A}^{(1:N)} = \tilde{\varpi}(F_{\tilde{\varpi}}^{-1}(U^{1:N}))$, where
  $U^j := \dfrac{j - 1 + \bar{U}}{N}$ for $j \in [N]$ 
	\State Draw $\bar{C} \sim \mathcal{U}([N^i])$ and set $C = \sigma_{-k}^N(\bar{C})$
	\State Set $A^{(j)} = \bar{A}^{(\sigma_{C}^N(j))}$, for $j \in [N]$
\State \textbf{output} $A^{(1:N)}$
\end{algorithmic}
\end{algorithm}

\begin{lemma}
    \label{lem:conditional-killing} 
The following define valid conditional resamplings (Definition
\ref{def:valid-conditional-resampling}):
\begin{enumerate}[(i)]
\item \label{item:killing-valid} conditional killing 
  $\rho_{\mathrm{kill}}^{(i,k)}$ 
  of Algorithm \ref{alg:killing}, and
\item \label{item:systematic-valid} conditional systematic resampling with mean partition $\rho_{\mathrm{syst}}^{(i, k)}$  of Algorithm \ref{alg:systematic}.
\end{enumerate}
\end{lemma} 

\section{Bridge backward sampling} \label{sec:cpf-bbs}

The backward/ancestor sampling CPF \citep{whiteley-backwards-note,lindsten-jordan-schon} often has impressive performance even with small $N$ and large $T$ \citep{lee-singh-vihola}. The selection probabilities in the backward sampling step 
include the transition density $M_{k+1}(X_{k+1}^{(B_{k+1})}\mid X_k^{(i)})$. When $M_{k+1}$ is slowly mixing, this density is typically very small for all $i$ except for the ancestor $i=A_k^{(B_{k+1})}$, and therefore the backward sampling step essentially reduces to ancestor tracing.

We discuss next the conditional particle filter with bridge backward sampling (CPF-BBS), which is a generalisation of CPF with backward sampling (CPF-BS) \citep{whiteley-backwards-note} suitable for slowly mixing $M_k$. The backward sampling step is replaced by a `bridging` procedure which spans over multiple time steps, and requires tractable dynamics $\{M_k\}$, in the following sense:
\begin{assumption} 
	\label{a:conditionals} 
	Denote $M_{\ell:u}(x_{\ell:u}) \defeq \prod_{k=\ell+1}^u M_k(x_k\mid x_{k-1})$
	for any $1\le \ell < u \le T$. 
	Then, we are able to simulate from and evaluate the density of the
	conditional distribution of $x_\ell$ given $x_{\ell-1}$ and $x_u$:
	\[
		\bar{M}_{\ell}(x_{\ell}\mid x_{\ell-1},x_u)
			 \defeq \frac{\int M_{\ell-1:u}(x_{\ell-1:u})
				\ud x_{\ell+1:u-1}}{M_{u\mid \ell-1}(x_u\mid x_{\ell-1})}.
			\]
	We further assume that we are able to evaluate the conditional density
	of $x_u$ given $x_{\ell}$:
	\[
		M_{u\mid \ell}(x_u\mid x_\ell) \defeq \int
				M_{\ell:u}(x_\ell,z_{\ell+1:u-1},x_u)\ud z_{\ell+1:u-1}.
			\]
\end{assumption} 
Assumption \ref{a:conditionals} is restrictive, and satisfied for instance by linear-Gaussian $M_k$. Note, however, that $M_k$ need not necessarily correspond to the statistical model, but may be another `proposal' distributions, as long as the Feynman--Kac model \eqref{eq:target} corresponds to the smoothing distribution (cf.~discussion in Section \ref{sec:discussion}).

Algorithm \ref{alg:bbcpf} gives the pseudocode of the CPF-BBS algorithm. 
\begin{algorithm}
	\spacingset{1.0}
  \caption{\small \textsc{CPF-BBS}$(X_{1:T}^*,B_{1:T}; T_{1:L})$}
  \label{alg:bbcpf} 
  \begin{algorithmic}[1]
    \small
  \State $(\underline{X}_{1:T}, \underline{A}_{1:T-1}, \tilde{B}_T) \gets \mathrm{CPF}(X_{1:T}^*,B_{1:T})$ 
      and set $\tilde{X}_T^* \gets X_T^{(\tilde{B}_T)}$
      \For{$k=L,L-1,\ldots,2$} 
            \State $\ell \gets T_{k-1}$; $u \gets T_{k}$; $B_u^* \gets
            \tilde{B}_u$
            \State $B_{\ell:u-1}^* \gets 
            \text{\textsc{AncestorTrace}}(\underline{A}_{\ell:u-1}, B_u^*)$
  
      \State $(\tilde{X}_{\ell:u-1}^*,\tilde{B}_{\ell:u-1}) \gets
      \textsc{BridgeCPF}(\underline{\vec{X}}_\ell, B_{\ell:u-1}^*,
      X_{\ell+1:u}^{(B_{\ell+1:u}^*)}, \ell, u)$
  
      \EndFor
      \State \textbf{output} $(\tilde{X}_{1:T}^*,\tilde{B}_{1:T})$
  \end{algorithmic}
  \end{algorithm}
The first step (line 1) invokes the forward CPF (Algorithm \ref{alg:cpf-core}). The bridging procedure (line 5) that replaces the usual backward sampling step in CPF-BS, requires a fixed `blocking sequence' $1 = T_1 < \cdots < T_L = T$ that gives rise to the blocks $(T_{i-1}, \ T_{i}), \ i = 2, \ldots, L$. For each block, $\ell=T_{i - 1}$ and $u=T_{i}$ are referred to as the block lower and upper boundaries, respectively, and Algorithm \ref{alg:bridge} attempts to change the ancestor of state $X_u^{(B^\ast_u)}$ at time $\ell$ from $B_\ell^\ast$ to a different particle from the pool $\underline{\vec{X}}_l$. Success of this step hinges on  Algorithm \ref{alg:bridge} being able to generate particles that could equally well explain the future state $X_u^{(B^\ast_u)}$ being conditioned on (see line 8 of Algorithm \ref{alg:bridge}), for which the conditional densities of Assumption \ref{a:conditionals} are needed in its forward simulation procedure (lines 1--7). The new ancestor from the pool $\underline{\vec{X}}_l$ is then found by ancestor tracing (line 9). Success in this step relies on  an efficient resampling strategy (line 3) to avoid particle degeneracy so that many particles from $\underline{\vec{X}}_l$ can survive to line 8. The choice of the  blocking sequence is important and we devise a practical design choice procedure for this in Section \ref{sec:plu}.

\begin{algorithm}
	\spacingset{1.0}
\caption{\small \textsc{BridgeCPF}$(\underline{\vec{x}}_\ell,
  b_{\ell:u-1}^*,x_{\ell+1:u}^*, \ell, u)$}
\label{alg:bridge} 
\begin{algorithmic}[1]
  \small
          \State 
          $W_{\ell}^{(1:N)} \gets M_{u\mid \ell}(x_u^* \mid
          \underline{x}_\ell)^{\frac{1}{u-\ell}}$; 
          $\underline{\tilde{\vec{X}}}_\ell \gets \underline{\vec{x}}_\ell$
          \For{$v = \ell+1:u-1$} \label{line:bridge-start}
            \State \label{line:bridge-resample} 
            $\tilde{A}_{v-1}^{(1:N)} \gets 
            r^{(b_{v-1}^*,b_v^*)}\big(\uarg \bigmid
            G_{v-1}(
            \tilde{\vec{X}}_{v-1}^{(1:N)})
            W_{v-1}^{(1:N)}
            \big)$ 
            \State Draw $\tilde{X}_v^{(i)} \sim \bar{M}_{v}(\uarg \mid
            \tilde{X}_{v-1}^{(\tilde{A}_{v-1}^{(i)})},
            \tilde{X}_u^{*})$ for $i\neq b_v^*$ and set 
            $\tilde{X}_v^{(b_v^*)} =
            x_v^*$
            \State Set $\tilde{\vec{X}}_{v}^{(i)}
            \gets (X_{v-1}^{(\tilde{A}_{v-1}^{(i)})}, \tilde{X}_v^{(i)})$
            for $i\in[N]$.
            \State $W_v^{(1:N)} \gets
            W_{v-1}^{(\tilde{A}_{v-1}^{(1:N)})}$
          \EndFor \label{line:bridge-end}
    \State Draw
    $\tilde{B}_{u-1}\sim\Categ(\tilde{\omega}_{u-1}^{(1:N)})$ where 
        $\tilde{\omega}_{u-1}^{(j)} = 
        G_{u-1}(
\tilde{\vec{X}}_{u-1}^{(j)})
        G_u(\tilde{X}_{u-1}^{(j)}, x_u^*)
        W_{u-1}^{(j)}$
     \State $\tilde{B}_{\ell:u-2} \gets 
          \text{\textsc{AncestorTrace}}(\underline{\tilde{A}}_{\ell:u-2}, \tilde{B}_{u-1})$
          \label{line:bridge-trace}
    \State \textbf{output}
    $\big((x_\ell^{(\tilde{B}_\ell)},\tilde{X}_{\ell+1:u-1}^{(\tilde{B}_{\ell+1:u-1})}),
\tilde{B}_{\ell:u-1}\big)$
\end{algorithmic}
\end{algorithm}

The following result, whose proof is given in Appendix \appref{app:bbcpf}, ensures the validity of CPF-BBS.

\begin{theorem} 
    \label{thm:bbcpf-valid} 
Assume that $N\ge 2$, $r^{(i,k)}$ is a valid conditional resampling (Definition \ref{def:valid-conditional-resampling}) and Assumption \ref{a:conditionals} holds.
Then, the Markov update
$(X_{1:T}^*, B_{1:T})\to (\tilde{X}_{1:T}, \tilde{B}_{1:T})$
defined by Algorithm \ref{alg:bbcpf} leaves $\pi\times \mathcal{U}([N]^T)$ invariant.
\end{theorem} 

With dense blocking sequence $T_{1:T} = 1{:}T$,
the bridging CPF and its tracing (lines \ref{line:bridge-start}--\ref{line:bridge-end} and \ref{line:bridge-trace} of Algorithm \ref{alg:bridge}, respectively) are eliminated, and therefore the CPF-BBS simplifies to the backward sampling CPF (CPF-BS) of \cite{whiteley-backwards-note}. This means that the CPF-BBS can be viewed as a true generalisation of CPF-BS for arbitrary blockings.

The other extreme case, that is, the trivial blocking sequence $T_1=1$, $T_2=T$ leads to running a CPF and then \emph{another} CPF 
with same initial particles and targeting the conditional distribution $\pi(x_{1:T-1}\mid \tilde{X}_T^*)$ (cf.~Lemma \sref{lem:bbcpf-recurse}). This may not be practically useful, but can give insight about what the `bridge CPF' is about. 

We conclude this section with two remarks about methods related to
Algorithm \ref{alg:bbcpf}.
\begin{enumerate}[(i)]
\item If we modify the algorithm by replacing \textsc{BridgeCPF} by the following algorithm:
  \begin{enumerate}[\small 1:]
    \item Set $\tilde{X}_{k}^{(B_\ell^*)} = X_{k}^{(B_{k}^*)}$ for $k=\ell{:}(u-1)$.
    \item For
    $i\neq B_\ell^*$, set $\tilde{X}_\ell^{(i)}= X_\ell^{(i)}$ and $\tilde{X}_{k}^{(i)}\sim \bar{M}_{v}(\uarg\mid \tilde{X}_{k-1}^{(i)}, X_u^{(B_u^*)})$ for $k=(\ell+1){:}(u-1)$. 
  \item Choose
    $\tilde{X}_{\ell:u-1}^{(i)}$ with probability proportional to
$M_{u\mid \ell}(\tilde{X}_u^*\mid \tilde{X}_\ell^{(i)})
      \prod_{v=\ell}^{u-1} G_v(\tilde{X}_v^{(i)})$,
\end{enumerate}
then we get a CPF version of the extended importance sampling for particle filters suggested in \citep{doucet-briers-senecal}.
We did not investigate such 'importance bridging' method further, but we believe that our bridge CPF generally allows for longer block sizes, for similar reasons why particle filters tend to be more efficient than importance sampling.
\item Algorithm \ref{alg:bbcpf} has similarities with the blocked particle Gibbs (or blocked CPF) of \cite{singh-lindsten-moulines}, but we believe that CPF-BBS can be substantially more efficient with the same computational complexity, because: 
\begin{itemize}
    \item CPF-BBS uses a block-wide `lookahead' which is possible
      to implement thanks to Assumption \ref{a:conditionals}, instead of using a
      modified potential only at the last time instant like (the direct implementation of) blocked particle Gibbs.
    \item The BBS update is not conditioned on a \emph{single} point at the block start, like the blocked particle Gibbs, but uses \emph{all particles} which were generated by the `forward' CPF.  
      (Algorithm \ref{alg:cpf-core}).
\end{itemize}
Note, however, that the blocked particle Gibbs is directly parallelisable unlike the CPF-BBS.
\end{enumerate}
Finally, we note that the reverse update order of the blocks occurs since a block's update depends on the value at the lower boundary of the subsequent block. 
Although not pursued here, it might also be possible to devise a forward only implementation, following \cite{lindsten-jordan-schon}. 

\section{Blocking sequence selection} \label{sec:blocking-seq-selection} 

The CPF-BBS (Algorithm \ref{alg:bbcpf}) is valid with any choice of the blocking sequence $T_{1:L}$.
However, its choice affects simulation efficiency, that is, the mixing of the
Markov chain.
In this section, we discuss a computationally inexpensive method that can be used in practice
to determine a suitable blocking sequence 
prior to running the CPF-BBS in order to facilitate efficient mixing.

We begin in Section \ref{sec:plu} by discussing a proxy for the integrated autocorrelation time (IACT) of the Markov chain output by the CPF-BBS. 
Then, Section \ref{sec:plu-estimator} details an estimator we have developed for the proxy.
Finally, Section \ref{sec:blocking-alg} describes a practical algorithm for blocking sequence selection that is based on the estimator of Section \ref{sec:plu-estimator}. 
We will study the methods presented in this section empirically in Section \ref{sec:experiments}. 

\subsection{The probability of lower boundary updates (PLU)} 
\label{sec:plu}

A theoretically attractive candidate strategy for blocking sequence selection is monitoring the IACT for variables of interest, based on the output of the CPF-BBS. 
Efficient inference could then be obtained by choosing the blocking sequence that minimises the IACT.
However, this approach is typically computationally demanding or even infeasible, since the estimation of the IACT is notoriously difficult and often requires extensive simulation of Markov chains. 

For these reasons, we base the selection of the blocking sequence on a proxy for IACT that is easier to work with. 
We call the proxy the `probability of lower boundary updates' ($\plu$),
and its definition for the block $(\ell, u)$, using the notation of Algorithm \ref{alg:bridge}, is:
\begin{equation}
  \plu(\ell, u) \defeq \mathbb{P}(X_\ell^{(\tilde{B}_{\ell})} \neq X_{\ell}^{(b_{\ell}^{*})}).
\end{equation}
In other words, $\plu(\ell, u)$ measures 
the probability that the bridge CPF (Algorithm \ref{alg:bridge}) on block $(\ell, u)$
updates the value at the block lower boundary $\ell$. 
Intuitively, higher values of $\plu$ should be associated with lower IACT.
Indeed, our experiments in Section \ref{sec:experiments} indicate that maximising $\plu$ (with respect to the block size) appears to yield a block size that approximately minimises IACT. 

\subsection{Approximate estimator for the PLU} \label{sec:plu-estimator}

Even though $\plu(\ell, u)$ is much easier to estimate than IACT, it still requires iterating the CPF-BBS for each candidate blocking, which is computationally demanding. We have developed an estimator for $\plu(\ell,u)$ which avoids this, and is based on a single 'stationary' CPF state (the generated particles $X_{1:T}^{(1:N}$ and reference indices $B_{1:T}$), which is used for any block boundaries $\ell,u$. The practical algorithm postponed to Section \ref{sec:blocking-alg} will be based on this idea, but assumes further that such a stationary CPF state can be well approximated by an independent particle filter.

The estimator from a single CPF state is presented in \eqref{eq:plu-est} below and is based on two `asymptotic' characterisations for $\plu$, 
for small and large block sizes, respectively.
The idea behind the characterisations is that the event $X_\ell^{(\tilde{B}_{\ell})} \neq X_{\ell}^{(b_{\ell}^{*})}$
occurs when a trajectory traced back from the generated particle tree in the bridge CPF has a different value at the block lower boundary than the reference. 

Consider first the case of a small block size, that is, $u - \ell \approx 1$.
In this case, $\plu$ is approximately characterised by:
\begin{equation}
  \label{eq:plu-m}
  \plu_{\mathrm{M}}(\ell, u) \defeq 1 - \dfrac{M_{u \mid \ell}(X_u^{*} \mid X_\ell^{*})}{\sum_{j = 1}^{N} M_{u \mid \ell}(X_u^{*} \mid X_\ell^{(j)})},
\end{equation}
where $X_\ell^* \defeq X_\ell^{(B_\ell)}$ and $X_u^* \defeq X_u^{(B_u)}$ refer to the $\ell$th and $u$th value of a reference trajectory.
The rationale for \eqref{eq:plu-m} comes from CPF-BS being a special case of the CPF-BBS for the dense blocking with unit block sizes.
Letting $b_{t}$ and $b_{t + 1}$ denote the indices of the current reference, the probability of choosing $b_{t}$ in backward sampling \citep{whiteley-backwards-note} is given by:
\begin{equation*} 
     \mathbb{P}(B_{t} = b_t \mid B_{t + 1} = b_{t + 1}) \propto 
     w_t^{(b_t)} 
     M_{t+1}(X_{t + 1}^{(b_{t+1})} \mid X_{t}^{(b_t)})
     G_{t+1}(X_{t}^{(b_t)}, X_{t + 1}^{(b_{t+1})}).
\end{equation*}
Here, under the weak potential setting (i.e.~with approximately constant potentials), the right hand side approximately reduces to $M_{u \mid \ell}(X_{u}^{(b_{u})} \mid X_{\ell}^{(b_\ell)})$ since $\ell = t, u = t + 1$ with a unit block size.
The probability of choosing a non-reference is therefore approximately given by \eqref{eq:plu-m}.

On the other hand, if the block size is large, $\plu(\ell, u)$ is approximately characterised by: 
\begin{equation}
  \label{eq:plu-g}
  \plu_{\mathrm{G}}(\ell, u) \defeq \Bigg (1 - \frac{1}{N} \Bigg) \prod_{k = \ell}^{u - 1} \Bigg( 1 - \frac{p_k N}{(N - 1)^{2}} \Bigg),
\end{equation}
where the quantity $p_k$ equals the probability that a resampling event occurs, divided by $N$. 
In the case of systematic resampling with mean partitioned weights $W_k^{(1:N)}$, (see Appendix A, Lemma 28 of \cite{chopin-singh-soto-vihola}), $p_k$ may be calculated as follows (for normalised $W_k^{(1:N)}$):
\begin{equation}
  \label{eq:cond-sys-p}
   p_k = \frac{1}{2} \sum_{i = 1}^{N} \Big | W_k^{(i)} - \frac{1}{N} \Big |.
\end{equation}
The justification of \eqref{eq:plu-g} comes from a calculation detailed in Appendix \appref{sec:plu-g-derivation}, which shows that 
$\plu_{\mathrm{G}}(\ell, u)$ approximately equals the expected proportion of particles whose ancestor at time $\ell$ is not the reference after an `artificial' conditional particle system has evolved for $u - \ell$ time steps from time $\ell$.
Therefore, $\plu_{\mathrm{G}}(\ell, u)$ may be loosely interpreted as approximating the probability of choosing non-reference at time $\ell$, when the ancestry of a particle chosen uniformly at time $u$ is traced back until time $\ell$.
 
Our estimator for $\plu(\ell, u)$ is constructed by `interpolating' \eqref{eq:plu-m} and \eqref{eq:plu-g} such that 
\begin{equation}
  \label{eq:plu-est}
  \widehat{\plu}(\ell, u) \defeq \plu_{\mathrm{G}}(\ell, u) \plu_{\mathrm{M}}(\ell, u) \Bigg ( 1 - \frac{1}{N} \Bigg)^{-1},
\end{equation}
where the scaling is added so that the estimator approximately reduces to 
\eqref{eq:plu-m} and \eqref{eq:plu-g} for short and long blocks, respectively, 
in the weak potential setting. 

The estimator in \eqref{eq:plu-est} was derived assuming an access to CPF state with $N$ particles. It is also possible to estimate the $\widehat{\plu}(\ell, u)$ from a CPF (or particle filter) state which has a different number of particles $N_0$ (which is often useful to take `large' in practice so that $N_0\gg N$). In this case, we can estimate $\plu_{\mathrm{G}}$ and $\plu_{\mathrm{M}}$ as follows, and then use \eqref{eq:plu-est} with the desired $N$ in the scaling.

To estimate $\plu_{\mathrm{G}}$, we simply compute $p_k$ using \eqref{eq:cond-sys-p} from the $N_0$ particles and 
substitute it directly to \eqref{eq:plu-g} with the desired $N < N_0$.
For $\plu_{\mathrm{M}}$ we use the alternative estimator of the form
\begin{equation}
  \label{eq:plu-m-alt}
  \plu_{\mathrm{M}}(\ell, u) = 1 - \dfrac{c(\ell, u)}{c(\ell, u) + N-1},
\end{equation}
which follows by assuming that 
\begin{equation}
  \label{eq:plu-m-alt-assump}
  M_{u \mid \ell}(X_u^{*} \mid X_{\ell}^{*}) \approx c(\ell, u) M_{u \mid \ell}^{(T)},
\end{equation}
where
\begin{equation}
  \label{eq:typical-dens-est}
  M_{u \mid \ell}^{(T)} = \dfrac{1}{N_0 - 1} \sum_{j \neq B_\ell} M_{u \mid \ell}(X_u^* \mid X_\ell^{(j)}).
\end{equation}
In other words, the block transition density for the reference is assumed to be approximately equal to a constant $c(\ell, u)$ times a `typical' value of the block transition densities for particles not including the reference. 
The estimator \eqref{eq:plu-m-alt} may be derived by appropriate substitution of \eqref{eq:plu-m-alt-assump} and \eqref{eq:typical-dens-est} into \eqref{eq:plu-m}.

\subsection{Algorithm for blocking sequence selection} \label{sec:blocking-alg}

In this section we describe a practical method based on \eqref{eq:plu-est} to choose the blocking sequence.
Algorithm \ref{alg:evaluate-blocking-candidates} describes a method that uses \eqref{eq:plu-est} to evaluate $S$ candidate blocking sequences $(T_{1:L^{(s)}}^{(s)})_{s = 1, 2, \dots, S}$ in the context of the FK distribution $(M_{1:T}, G_{1:T})$. 
\begin{algorithm}
	\spacingset{1.0}
  \caption{\small \textsc{EvaluateBlockingCandidates}($\{T_{1:L^{(1)}}^{(1)}, \dots, T_{1:L^{(S)}}^{(S)}\}$, $M_{1:T}$, $G_{1:T}$, $N$, $n$)}
  \label{alg:evaluate-blocking-candidates}
  \begin{algorithmic}[1]
    \small
    \For{$j = 1, \dots, n$}
      \State $\underline{X}_{1:T}, \underline{A}_{1:T-1}, \underline{W}_{1:T} \gets \textsc{PF}(\rho_{\mathrm{syst}}, M_{1:T}, G_{1:T}, N)$ \label{line:pf-start} 
      \State Draw $B_T \sim \Categ(W_T^{(1:N)})$; 
      $B_{1:T-1} \gets \textsc{AncestorTrace}(A_{1:T-1}^{(1:N)}, B_T)$ \label{line:pf-end}
      \State $\phi_{\mathrm{PLU}}[:, :, j] \gets \textsc{EstimatePLU}(\{T_{1:L^{(1)}}^{(1)}, \dots, T_{1:L^{(S)}}^{(S)}\}$, $\underline{X}_{1:T}, \underline{W}_{1:T-1}, B_{1:T})$ \label{line:compute-block-effs}
          \EndFor 
    \For{$s = 1, \dots, S$} \label{line:summ-start}
      \State Set $\bar{\phi}_{\mathrm{PLU}}[i, s] = \textsc{Mean}(\phi_{\mathrm{PLU}}[i, s, :])$ for $i = 1, \dots, L^{(s)} - 1$.
    \EndFor \label{line:summ-end} \\ 
    \Return $\bar{\phi}_{\mathrm{PLU}}$ 
  \end{algorithmic}
\end{algorithm}
The additional parameters $N$ and $n$ stand for the number of particles and number of iterations, which are tuning parameters of the blocking candidate evaluation.
Here, we use indexing notation where $A[i, j, k]$ stands for the element in row $i$, column $j$ and slice $k$ in an array $A$.
Furthermore, the columns of arrays need not have the same number of rows, and indexing operations with `$:$' mean `all elements' in the particular dimension. 

One iteration of the main loop in Algorithm \ref{alg:evaluate-blocking-candidates} consists of running the standard particle filter (Algorithm \ref{alg:pf}) with mean partitioned systematic resampling followed by a traceback using ancestor tracing (Algorithm \ref{alg:at}) in lines \ref{line:pf-start}--\ref{line:pf-end}.
Then, given the output of the particle filter, we estimate $\plu$ using Algorithm \ref{alg:estimate-plu} (see below) on line \ref{line:compute-block-effs} for each block $(\ell, u)$ within each blocking sequence.
The computation is a straightforward application of Equations \eqref{eq:plu-m}--\eqref{eq:plu-est} using the particle filtering results.
Finally, lines \ref{line:summ-start}--\ref{line:summ-end} summarise the estimates of $\plu$ by taking their mean over the $n$ replicate runs of the particle filter.
The element $\bar{\phi}_{\mathrm{PLU}}[i, s]$ in the output of Algorithm \ref{alg:evaluate-blocking-candidates} describes in terms of $\plu$, how efficient the $i$th block in the blocking sequence $s$ was.

\begin{algorithm}
	\spacingset{1.0}
  \caption{\small \textsc{EstimatePLU}($\{T_{1:L^{(1)}}^{(1)}, \dots, T_{1:L^{(S)}}^{(S)}\}$, $\underline{X}_{1:T}, \underline{W}_{1:T-1}, B_{1:T}$)}
  \label{alg:estimate-plu}
  \begin{algorithmic}[1]
    \small
    \State Compute $p_k$ for $k = 1, \dots, T-1$ using \eqref{eq:cond-sys-p}.
  \For{$s = 1, \dots, S$}
        \For{$i = 1, \dots, L^{(s)} - 1$}
          \State Set $\ell = T_{i}^{(s)};$ 
          $u = T_{i + 1}^{(s)};$
          $X_{\ell}^{*} = X_{\ell}^{(B_\ell)};$ 
          $X_{u}^{*} = X_{u}^{(B_u)}$
          \State Compute $\plu_M(\ell, u)$ using \eqref{eq:plu-m} and
          $\plu_G(\ell, u)$ using \eqref{eq:plu-g} 
          \State Set $\phi_{\mathrm{PLU}}[i, s] = \widehat{\plu}(\ell, u)$
          given in \eqref{eq:plu-est}
        \EndFor
  \EndFor \\
    \Return $\phi_{\mathrm{PLU}}$
  \end{algorithmic}
\end{algorithm}

Algorithm \ref{alg:evaluate-blocking-candidates} may in principle be used to evaluate any candidate blocking sequence, but we suggest to use it with Algorithm \sref{alg:dyadic-candidate-blockings} given in Appendix \appref{sec:construct-dyadic-blockings} that constructs dyadic blocking sequences: the block sizes $T_{k+1} - T_k$ are powers of two. 
More precisely, if $T = 2^{p^{*}} + 1$ for some $p^{*}$, Algorithm \sref{alg:dyadic-candidate-blockings} returns blocking sequences $T^{(i)}_{1:L^{(i)}}$ for $i = 1, 2, \dots, p^{*} + 1$,
where the block sizes of the $i$th sequence are all constant $2^{i-1}$, except for a possible `residual block' of length $< 2^{i-1}$ as the last block in each sequence $i$.

Finally, Algorithm \ref{alg:choose-blocking} describes a method based on Algorithms \ref{alg:evaluate-blocking-candidates} and \sref{alg:dyadic-candidate-blockings}
for choosing a single blocking sequence to be used with the CPF-BBS and a given FK distribution. 
In summary, Algorithm \ref{alg:choose-blocking} first constructs the candidate blocking sequences using Algorithm \sref{alg:dyadic-candidate-blockings}. 
Then, Algorithm \ref{alg:evaluate-blocking-candidates} is run to obtain $\bar{\phi}_{\mathrm{PLU}}$ given these sequences.
The data $\bar{\phi}_{\mathrm{PLU}}$ is then reinterpreted as a set of elements $D_{\mathrm{PLU}}$, whose element $(\ell, b, e_{\plu})$
describes the estimated $\plu$, $e_{\plu}$, of the block with lower boundary $\ell$ and upper boundary $\ell + b$. 
Finally, $D_{\mathrm{PLU}}$ is processed such that blocking sequences with largest block sizes are considered first, and at each block lower boundary, the best performing block size in terms of the estimated $\plu$ is selected to the output blocking sequence.
\begin{algorithm}
	\spacingset{1.0}
  \caption{\small \textsc{ChooseBlocking}($M_{1:T}, G_{1:T}, N, n)$}
  \label{alg:choose-blocking}
  \begin{algorithmic}[1]
    \small
    \State $\{T_{1:{L^{(1)}}}^{(1)}, \ldots, T_{1:{L^{(p)}}}^{(p)} \} \gets \textsc{DyadicCandidateBlockings}(T)$
    \State $\bar{\phi}_{\mathrm{PLU}} \gets \textsc{EvaluateBlockingCandidates}(\{T_{1:{L^{(1)}}}^{(1)}, \ldots, T_{1:{L^{(p)}}}^{(p)} \}, M_{1:T}, G_{1:T}, N, n)$
	  \State Compute $D_{\mathrm{PLU}}$, a container with elements of the form $(\ell, b, e_{\plu})$ based on $\bar{\phi}_{\mathrm{PLU}}$.
    \State Initialise $D$, an empty container for elements of the form $(\ell, b)$. 
    \For{$s = p, p - 1, \ldots, 1$}
       \State Get lower boundaries and block sizes $(\ell_k, b_k)$ for $k = 1, \ldots, L^{(s)} - 1$ from $T_{1:L^{(s)}}^{(s)}$.
      \For{$k = 1, \ldots, L^{(s)} - 1$}
        \State Denote by $D_{\mathrm{PLU}}^{(\ell_k)}$ all elements of $D_{\mathrm{PLU}}$ whose block lower boundary equals $\ell_k$.
		  \If{maximal $e_{\mathrm{PLU}}$ is reached when block size equals $b_k$ among elements of $D_{\mathrm{PLU}}^{(\ell_k)}$}
           \State Add $(\ell_k, b_k)$ to $D$.
           \State Remove all elements of $D_{\mathrm{PLU}}$ with $\ell$ such that $\ell_k \leq \ell < \ell_k + b_k$.
        \EndIf
      \EndFor
    \EndFor \\
    \Return Blocking sequence constructed from elements of $D$.
  \end{algorithmic}
\end{algorithm}

\section{Linear diffusions with path integral weights} \label{sec:fk-path-integral} 

We discuss next a class of continuous-time models and their discretisations, for which the methods of Section \ref{sec:cpf-bbs}--\ref{sec:blocking-seq-selection} are particularly useful. We will consider instances of these models also in the experiments (Section \ref{sec:experiments}).

We start with the continuous-time model on a time interval $[0,\tau]$. The prior dynamics $\mathbb{M}$ correspond to the solution of a linear stochastic differential equation (SDE):
\begin{equation}
  \label{eq:lin-sde}
  \ud X_t = \mathbf{F} X_t \ud t + \mathbf{K} \ud \mathbf{B}_t,
  \qquad X_0 \sim N(\mu_\mathrm{init}, \Sigma_\mathrm{init})
\end{equation}
where $\mathbf{B}_t$ is a $d$-dimensional Brownian motion and $\mathbf{F}$ and $\mathbf{K}$ are matrices of appropriate dimension, and $\mu_\mathrm{init}$ and $\Sigma_\mathrm{init}$ are the mean and covariance of the initial distribution, respectively. The law of interest is $\mathbb{M}$ weighted by non-negative weights of the form
$w(x_{[0,\tau]}) = \exp(- \int_{0}^{\tau} \mathcal{V}(x_u) \ud u )$,
where $\mathcal{V}:\X\to[0,\infty]$ are `potential' functions that `penalise' the trajectories of $\mathbb{M}$. That is, the distribution of interest is proportional to $\mathbb{M}(\ud x_{[0,\tau]}) w(x_{[0,\tau]})$.

In practice, we assume a time discretisation of $[0,\tau]$, $0 = t_1 < t_2 < \cdots < t_T = \tau$,
which leads to the discrete-time FK distribution \eqref{eq:target}.
The dynamics $M_{1:T}$ in \eqref{eq:target} correspond to the marginals of $X_{[0,\tau]}\sim \mathbb{M}$, that is:
\begin{equation}
  \label{eq:fk-path-integral-M}
  \begin{aligned}
    M_1 &= \mathrm{Law}(X_{t_1}) = N(\mu_\mathrm{init},\Sigma_\mathrm{init}) \\ 
    M_k(\uarg\mid x) &= \mathrm{Law}(X_{t_k} \mid X_{t_{k-1}}=x) \qquad \text{for } 2 \leq k \leq T,
  \end{aligned}
\end{equation}
which are linear-Gaussian. Appendix \appref{sec:lgssm-conditionals} details how $M_k$ can be derived from the parameters of the SDE, and also how their necessary conditional distributions required by Assumption \ref{a:conditionals} can be determined. 
The potential functions $G_{1:T}$ in \eqref{eq:target} stem from approximating the path integral by a Riemann sum:
\begin{equation}
  \label{eq:path-integral-approx}
  w(x_{[0,\tau]}) = 
  \prod_{k=1}^{T-1} \exp{\Bigg( - \int_{t_k}^{t_{k+1}} \mathcal{V}(x_u) \ud u \Bigg)}
  \approx \prod_{k=1}^{T-1} \exp{\Bigg( -|\Delta_k| \mathcal{V}(x_{t_k}) \Bigg)},
\end{equation}
where $\Delta_k = [t_k, t_{k+1})$ and $|\Delta_k| = t_{k+1} - t_k$.
This leads to potentials of the following form:
\begin{equation}
  \label{eq:fk-path-integral-G}
  \begin{aligned}
    G_1(x_{t_1}) &= \exp\big(-(t_2 - t_1) \mathcal{V}(x_{t_1}) \big) \\
    G_k(x_{t_{k-1}}, x_{t_{k}}) &= \exp\big(-( t_{k + 1}-t_{k})  \mathcal{V}(x_{t_{k}})\big) \text{ for } 2 \leq k \leq T - 1, \text{ and } 
    G_T \equiv 1.
  \end{aligned}
\end{equation}
\begin{remark}
  The scenario detailed above can be generalised and/or modified in a number of ways. Indeed, the potentials $G_k$ can also include purely discrete-time elements, as in our Cox process experiment (Section \ref{sec:experiments-blocking}).
  The law $\mathbb{M}$, or equivalently $M_k$, can also correspond to the law of linear SDE \emph{conditioned on} a number of linear-Gaussian observations. In such a case, the distributions $M_k$ are still linear-Gaussian, and we can derive the required conditional laws. 
  This can be useful in many practical settings, and indeed was essential for our movement model example (Section \ref{sec:movement-modelling}).
\end{remark}
\section{Experiments} \label{sec:experiments} 

\subsection{Comparison of conditional resamplings} \label{sec:comparing-cond-res}

We first investigate the performance of the CPF-BBS (Algorithm \ref{alg:bbcpf}) using the conditional resamplings $\rho_{\mathrm{kill}}$ and $\rho_{\mathrm{syst}}$. 
For reference, we also study conditional multinomial resampling with conditioning indices $i$ and $k$, $\rho_{\mathrm{mult}}^{(i, k)}$.
This conditional resampling may be simply implemented by first drawing the 
ancestor indices $A^{(1:N)} \sim \Categ(w^{(1:N)})$ as in standard multinomial resampling, and then enforcing the condition $A^{(k)} = i$ (since $A^{(1:N)}$ are independent).

In this section, we study a correlated random walk incorporating a path integral type potential function, hereafter called the CTCRW-P model.
The dynamics of the model $X_t = (V_t \ L_t)^T$ are driven by the SDE
\begin{equation}
  \label{eq:ctcrwp-sde}
  \begin{aligned}
    \ud V_t &= -\beta_v V_t \ud t + \sigma \ud B_t \\
    \ud L_t &= [-\beta_x L_t + V_t] \ud t,
  \end{aligned}
\end{equation}
where $B_t$ is the standard Brownian motion, $\sigma$, $\beta_v$ and $\beta_x$ are parameters, and $(L_t)_{t \geq 0}$ and $(V_t)_{t \geq 0}$ represent location and velocity processes, respectively.
The FK representation \eqref{eq:fk-path-integral-M} \& \eqref{eq:fk-path-integral-G} of CTCRW-P is given by
$M_1 := N(0, S)$, $M_k(\uarg \mid x) := N(T_{t_{k - 1}, t_{k}} x, Q_{t_{k-1}, t_{k}}), \ \text{for } 2 \leq k \leq T$ and $\mathcal{V}(X_t) := L_{t}^2 / (2\eta^2)$.
Here, $\eta$ is a parameter, and $T_{t_{k-1}, t_{k}}$, $Q_{t_{k-1}, t_{k}}$ and $S$ are the transition matrix, conditional covariance matrix and stationary covariance matrix, respectively, arising in the solution of the linear SDE \eqref{eq:ctcrwp-sde}.
Their expressions are given in Appendix \appref{sec:ctcrwp-details}, in Equations \seqref{eq:ctcrwp-expm}, \seqref{eq:upper-tri-ctcrwp}--\seqref{eq:upper-tri-ctcrwp-eq} and \seqref{eq:ctcrwp-statcov}--\seqref{eq:ctcrwp-statcov-eq}, respectively.

We ran the CPF-BBS targeting CTCRW-P with the configurations 
$N \in \{2, 4, 8, 16, 32\}$, $\blocktime \in \{2^{-7}, 2^{-6}, \dots 2^6\}$ and  
$r \in \{ \rho_{\mathrm{syst}}, \rho_{\mathrm{kill}}, \rho_{\mathrm{mult}} \}$.
Here, $\blocktime$ parameterises the blocking sequence in terms of the `physical time' of the discretised SDE.
The blocksizes $T_{k + 1} - T_k$ in Algorithm \ref{alg:bbcpf}, may simply be obtained by dividing $\blocktime$ by $|\Delta_k|$ (see below).
For each run of the CPF-BBS, we used 21000 iterations with the first 1000 discarded as burn-in.

We set $\tau = 2^6$, $|\Delta_k| = 2^{-7}$, $\eta = 1.0$ and $\sigma \in \{0.125, 0.5, 2.0\}$, which controls the variability in the velocity process.
Each time, given $\sigma$, we solved for the parameters $\beta_x$ and $\beta_v$ such that the stationary covariance matrix \seqref{eq:ctcrwp-statcov} had unit variances on the diagonal.
This was done to ensure that the variability of the process remains similar as $\sigma$ changes.

The simulations were run with all combinations of the algorithm and model configurations described above.
We estimated PLU (discussed in Section \ref{sec:plu}) by tallying iterations where 
$x_\ell^{(\tilde{B}_{\ell})} \neq x_{\ell}^{(b_{\ell}^{*})}$ and dividing by their total,
and estimated the $\iact$ for $L_{0.0}$ using batch means \citep{flegal-batch}.
\begin{figure}
  \centering
  \includegraphics{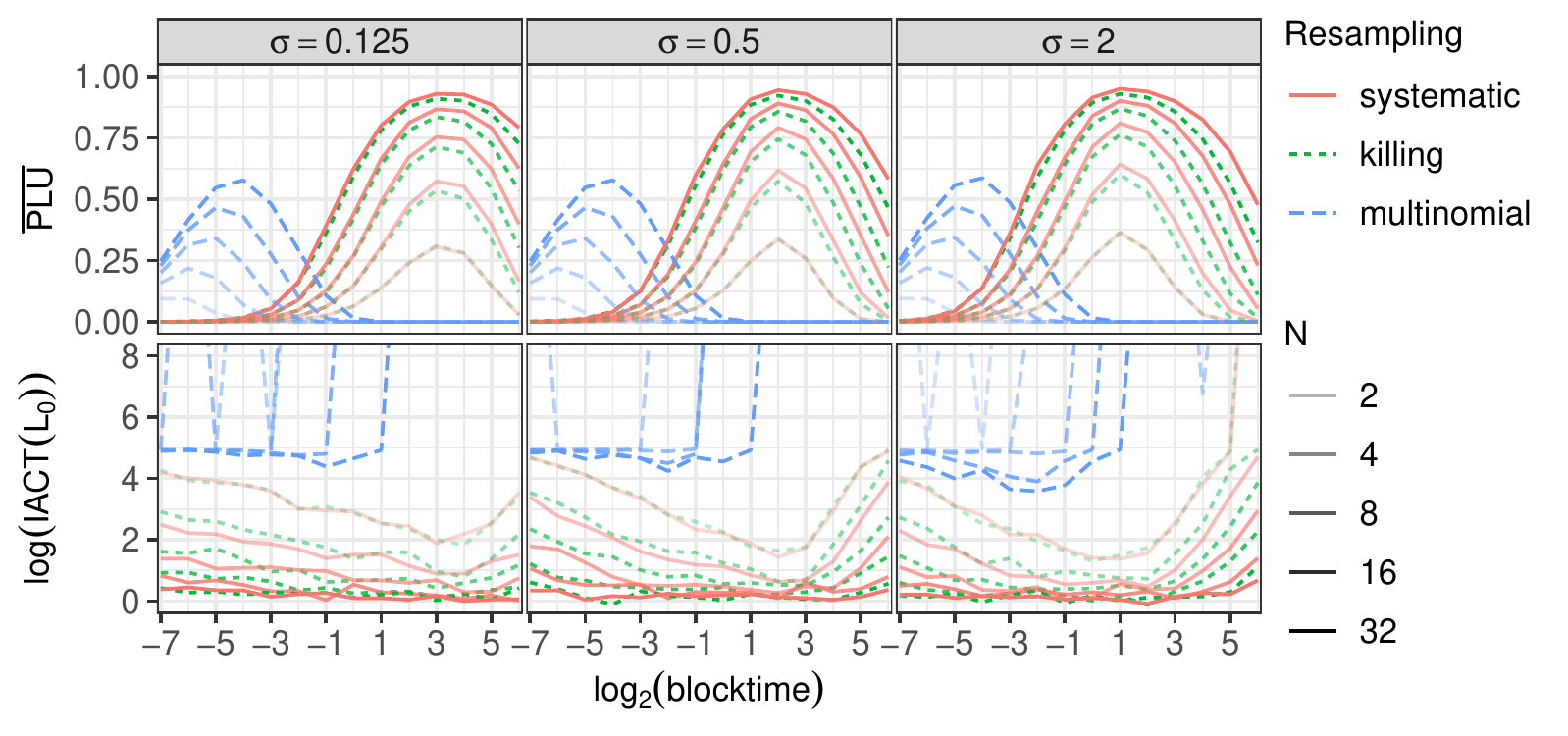}
  \caption{\small The estimated mean PLUs and the logarithm of IACT with varying $\sigma$ for the location state variable at time 
  $0.0$ in the CTCRW-P model. The value of $|\Delta_k|$ was set to $2^{-7}$. The performance of CPF-BS is seen at the far left, with $\blocktime = 2^{-7}$.}
  \label{fig:ctcrwp-param-vs-blocksize}
\end{figure} 
Figure \ref{fig:ctcrwp-param-vs-blocksize} summarises the results of this experiment. 
The mean $\plu$ shown in the top row is computed over the number of blocks (given here by $\tau / \blocktime$).
The figure shows systematic and killing resampling performing better than multinomial resampling,
which can be seen from the lower IACTs and higher mean PLU.
The performance with multinomial resampling is poor here, as expected, since the model has weak potentials with $|\Delta_k| = 2^{-7}$.
In contrast, killing and systematic resampling behave nearly uniformly, with systematic resampling performing slightly better. 
This finding aligns well with the theoretical and empirical findings in \citep{chopin-singh-soto-vihola} for the particle filter in a similar context of path integral potentials and $|\Delta_k|$ close to $0$. 

The CPF-BBS coincides with the CPF-BS when $\blocktime = |\Delta_k|$, which corresponds to the first value on the horizontal axis.
Even though increasing $N$ naturally improves the performance of the CPF-BS too, the CPF-BBS has better simulation efficiency with an appropriately
chosen blocktime, for any $N$ in the simulation. 
Note that the estimation of the IACT is quite noisy here, since the mixing is poor especially with multinomial resampling
and with poorly chosen blocking sequences induced by the value of $\blocktime$.
In contrast, the computed mean $\plu$ appears less noisy, and in the case of systematic and killing resampling the best blocktime in terms of $\iact$ is identified.

We also investigated the relationship of $\plu$ with $\iact_{32.0}$, and the findings were similar. 
A further experiment fixing $\sigma = 1.0$ and varying $\eta \in \{0.125, 0.5, 2.0\}$ instead also resulted in similar findings (see supplementary Figure \sref{fig:ctcrwp-potential-vs-blocksize-supp}). 

\subsection{Choice of the blocking sequence}
\label{sec:experiments-blocking}

As already illustrated empirically with Figure \ref{fig:ctcrwp-param-vs-blocksize} and discussed in Section \ref{sec:blocking-seq-selection}, the choice of the blocking sequence is a tuning parameter affecting the sampling efficiency of the CPF-BBS.
Figure \ref{fig:ctcrwp-sigma-iact-all-times} exemplifies this further by showing another look at the results obtained from the
experiment in the previous section. 
Here, the logarithm of the inverse relative efficiency (IRE) is plotted at each time point when systematic resampling was used.
The IRE is obtained by scaling the IACT by the number of particles, and measures the asymptotic efficiencies of estimators with varying computational costs \citep{glynn1992asymptotic}. 
The panes from left to right show the results with odd $\blocktime$ values and represent 
a range of algorithms beginning from the CPF-BS ($\blocktime = 2^{-7}$).
Here, the algorithms that are tuned well use only 4 particles, motivating the search for an appropriate blocktime (or blocking sequence).
By visual inspection, it appears that the best $\blocktime$ values found in the experiment are roughly $2^1$ for $\sigma = 2.0$, 
$2^3$ for $\sigma = 0.5$ and $2^3 - 2^5$ for $\sigma = 0.125$. 
When the value of $\sigma$ is decreased, larger $\blocktime$ achieves better performance, since a small $\sigma$
leads to `stiff' dynamics $M_{1:T}$. 
\begin{figure}
  \centering
  \includegraphics{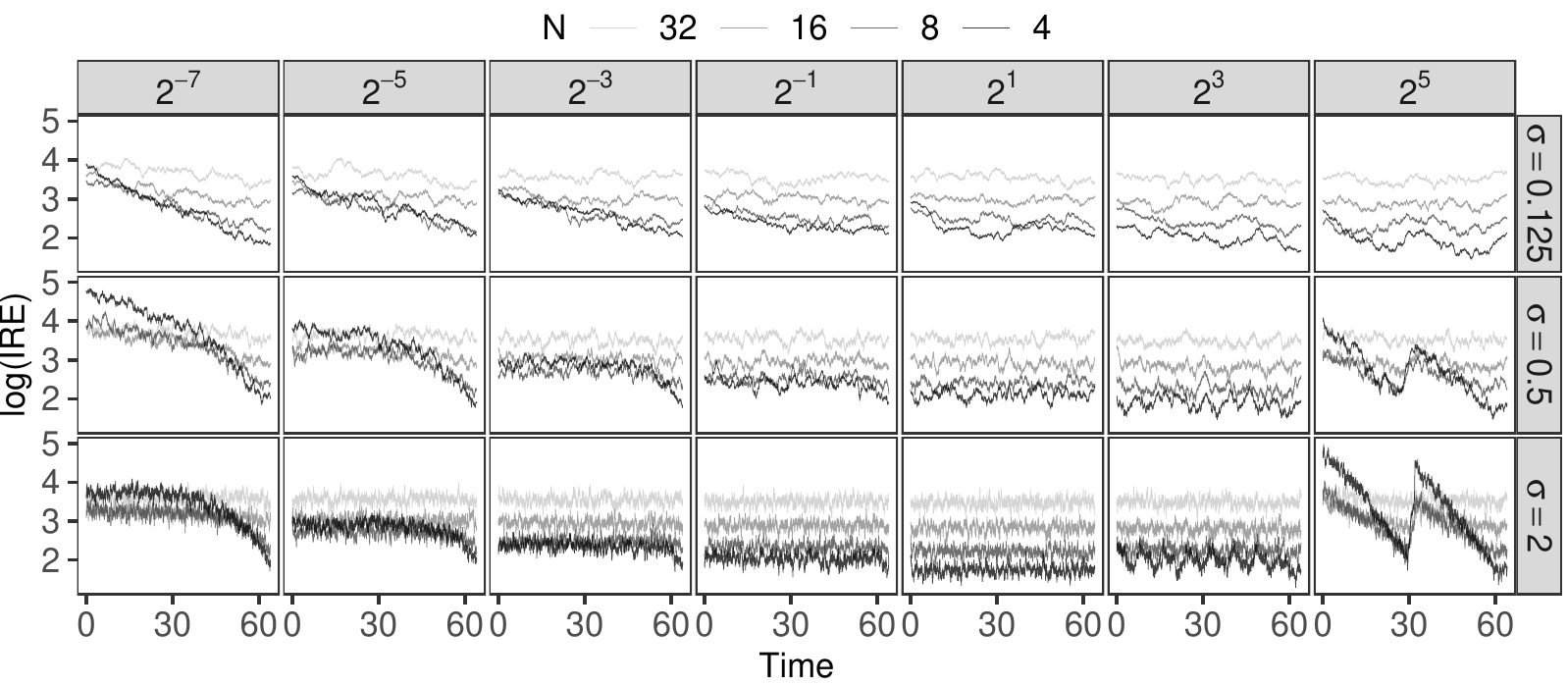}
  \caption{\small The logarithms of inverse relative efficiency obtained at each time point with conditional systematic resampling in the experiment discussed in Section \ref{sec:comparing-cond-res}. The columns show the results with varying $\blocktime$ in Algorithm \ref{alg:bbcpf}.}
  \label{fig:ctcrwp-sigma-iact-all-times}
\end{figure}

The best found blocktimes represent balances between two phenomena.
On the one hand, the blocks are large enough so that the dynamic model has sufficient time to bridge from the block lower boundaries to the values conditioned on at the upper boundaries (in Algorithm \ref{alg:bridge}). 
On the other hand, the blocks are small enough to avoid particle degeneracy within the blocks. 

Next, we investigate how well the estimates of $\bar{\phi}_{\plu}$ computed using Algorithm \ref{alg:evaluate-blocking-candidates} coincide with $\plu$.
We studied the relationship of $\bar{\phi}_{\mathrm{PLU}}$ and $\plu$ with respect to $\blocktime$ (that is, with blocking sequences constructed with constant block sizes) using the CTCRW-P model with $N \in \{2^1, 2^2, \ldots, 2^{10}\}$ and the parameter $\sigma \in \{0.03125, 0.125, 0.5, 2, 4\}$. 
The rest of the model configuration was as in Section \ref{sec:comparing-cond-res}.
To estimate $\plu$, 
we ran 1100 iterations of Algorithm \ref{alg:bbcpf} with the first 100 discarded burn-in, monitoring for each block the proportion of iterations where $x_\ell^{(\tilde{B}_{\ell})} \neq x_{\ell}^{(b_{\ell}^{*})}$.
In Algorithm \ref{alg:evaluate-blocking-candidates}, we used $n = 50$ runs of the particle filter and $N$ as reported above.
Figure \ref{fig:blocksize-tuning-ctcrwp} visualises the results for $N \in \{2, 8, 32, 128\}$ (the results for other $N$ yield no further conclusions).
The estimated $\plu$ and $\bar{\phi}_{\plu}$ appear to be in close agreement, with only slight discrepancies seen for large blocktimes. 
\begin{figure}
  \centering
  \includegraphics{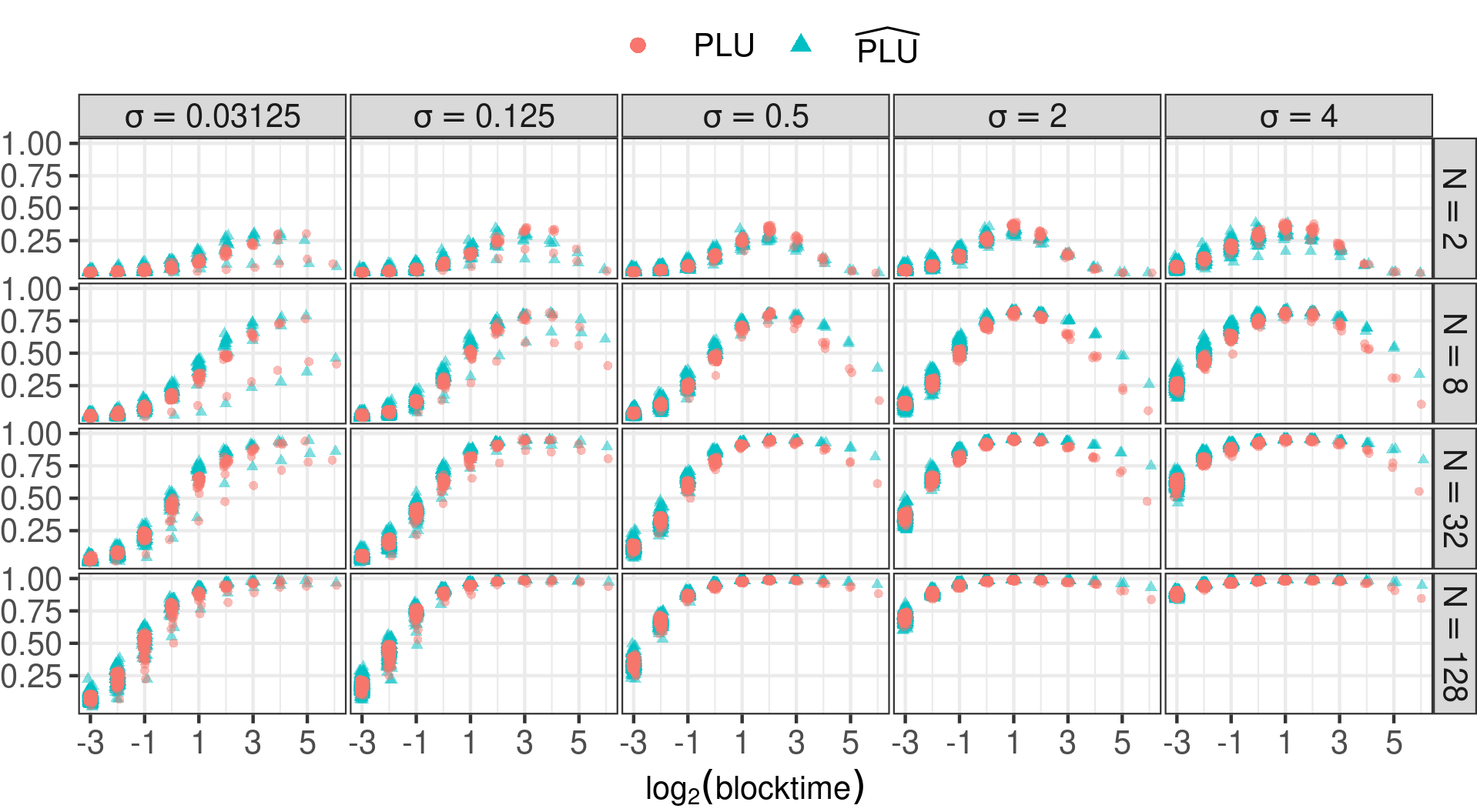}
  \caption{\small The $\plu$ (orange) and $\bar{\phi}_{\mathrm{PLU}}$ (light blue) for each block induced by the $\blocktime$ on the horizontal axis for the CTCRW-P. The points are slightly jittered for visualisation.}
  \label{fig:blocksize-tuning-ctcrwp}
\end{figure}
This finding motivates the use of $\bar{\phi}_{\mathrm{PLU}}$ as a maximisation criterion for finding a block size that 
likely results in a high overall $\plu$.

Next, we turned to study Algorithm \ref{alg:choose-blocking} for selecting the blocking sequence based on $\bar{\phi}_{\mathrm{PLU}}$. 
We investigated this with a model that slightly differs from the form \eqref{eq:fk-path-integral-G}, and is a Cox process model incorporating a reflected Brownian motion (CP-RBM) first appearing in \citep{chopin-singh-soto-vihola} and briefly detailed (with minor changes) below.

The CP-RBM model is a random intensity time-inhomogeneous Poisson process, generating observation sequences $\tilde{\tau}$. 
The intensity function is piecewise constant, and given by
\begin{equation}
  \label{eq:cp-rbm-intensity}
  \lambda(t) = \beta\exp{(-\alpha X_{t_{k}})}, \ \text{ for } t \in [t_{k}, t_{k + 1}),
\end{equation}
where $\alpha$ and $\beta$ are parameters.
The process $(X_{t_k})_{k = 1, \dots, T}$ is distributed such that 
\begin{equation}
    \label{eq:cp-rbm-dynamics}
    X_{t_1} \sim N^{(r)}(0, 1, a, b), \text{ and }
	 X_{t_k} \mid (X_{t_{k-1}} = x_{t_{k-1}}) \sim N^{(r)}(x_{t_{k-1}}, |\Delta_k|\sigma^2, a, b), \\
\end{equation}
where $N^{(r)}(\mu, \sigma ^ 2, a, b)$ is a distribution we call the `reflected normal distribution', with parameters $\mu$, $\sigma$ and bounds $a$ and $b$.
To simulate from $N^{(r)}(\mu, \sigma ^ 2, a, b)$, one first draws $Z \sim N(\mu, \sigma^2)$ and then sets $X = \mathrm{reflect}(Z, a, b)$,
where `reflect' is an operation that recursively reflects (that is, mirrors over a boundary) $Z$ with respect to $a$ (if $Z < a$) or $b$ (if $Z > b$) until a value within $(a, b)$ is obtained and outputted.

To apply the CPF-BBS with the CP-RBM, we use the following FK representation:
\begin{equation}
  \label{eq:cp-rbm-fk}
  \begin{aligned}
    M_1 &:= N(0, 1), \text{ and }
    M_k(\uarg \mid x) := N(x, |\Delta_{k -1}| \sigma^2), \ \text{for } 2 \leq k \leq T \\
    G_1(x) &:= \dfrac{N^{(r)}(x; 0, 1, a, b)\exp{(-|\Delta_1|\beta\exp{(-\alpha x)})}}{N(x; 0, 1)} 
    (\beta\exp{(-\alpha x)})^{\charfun{\exists i \text{ s.t } \tilde{\tau}_i \in \Delta_1}}
                     \\
    G_k(x, y) &:= \dfrac{N^{(r)}(y; x, |\Delta_{k-1}|\sigma^2, a, b)
    \exp{(-|\Delta_k|\beta\exp{(-\alpha y)})}}{N(y; x, |\Delta_{k- 1}|\sigma^2)}\times \\ 
    &\phantom{:=}(\beta\exp{(-\alpha y)})^{\charfun{\exists i \text{ s.t } \tilde{\tau}_i \in \Delta_k}},
  \ \text{for } 2 \leq k \leq T \\
  \end{aligned}
\end{equation}
where $|\Delta_{T}| = 0$. Note that these $M_k$ differ from \citep{chopin-singh-soto-vihola} and satisfy Assumption \ref{a:conditionals}, even though the statistical dynamic model \eqref{eq:cp-rbm-dynamics} does not. Here, the intractable reflection part of the model dynamics is accounted for in the potential functions $G_k$.
The FK distribution above is valid for the inference of the CP-RBM in the situation that the time discretisation is made fine enough such that each $\Delta_k$ contains at most one observation.
The density $N^{(r)}(x; \mu, \sigma^2, a, b)$ contains an infinite sum (see Appendix \appref{sec:cp-rbm-details}), which we truncate to the first ten terms. 

We first drew a realisation of the process $X$ using \eqref{eq:cp-rbm-dynamics} with $|\Delta_k| = 2^{-6}$, $\sigma = 0.3$, $a = 0$, $b = 3$ and time interval length $\tau = 2^8$.
Then, conditional on this realisation, we simulated one dataset, $\tilde{\tau}$, from the Poisson process with intensity \eqref{eq:cp-rbm-intensity} with $\alpha = 1$ and $\beta = 0.5$.
Finally, we augmented the discretisation of the process $X$ with the time points $\tilde{\tau}$, leading to a model of the form \eqref{eq:cp-rbm-fk}.

For the blocking sequences, we considered the sequences induced by the constant blocktimes $\{2^{-6}, 2^{-5}, \ldots, 2^5\}$ and 
a (inhomogeneous) blocking sequence constructed using Algorithm \ref{alg:choose-blocking} with $n = 50$ and $N = 8$.
Here, a minor change to the choice of candidate blockings (that is, Algorithm \sref{alg:dyadic-candidate-blockings}) was done: instead of constructing them using
block sizes (integers) in powers of two as discussed in Section \ref{sec:blocking-seq-selection}, we constructed them using the power of two block\textit{times} $2^{-6}-2^{5}$ as this is more natural for a continuous-time model.
For each blocking sequence, we then applied the CPF-BBS with $N = 8$ for 26000 iterations with the first 1000 discarded as burn-in.

Figure \ref{fig:cp-rbm-summary} summarises the results of the experiment.
The top pane shows the true simulated state, the observations $\tilde{\tau}$ and the 50\% and 95\% credible intervals of the distributions $X_t \mid \tilde{\tau}$. 
The middle pane compares the IACTs obtained from the samples of said distributions with some of the considered blocking strategies; the inhomogeneous blocking is highlighted in red. 
The blocking strategies omitted from the figure yield no further conclusions and the IACTs for the blocktimes $> 2^3$ were greater than for the strategies depicted.
Finally, the bottom pane visualises the inhomogeneous blocking sequence obtained using Algorithm \ref{alg:choose-blocking}.
\begin{figure}
  \centering
  \includegraphics{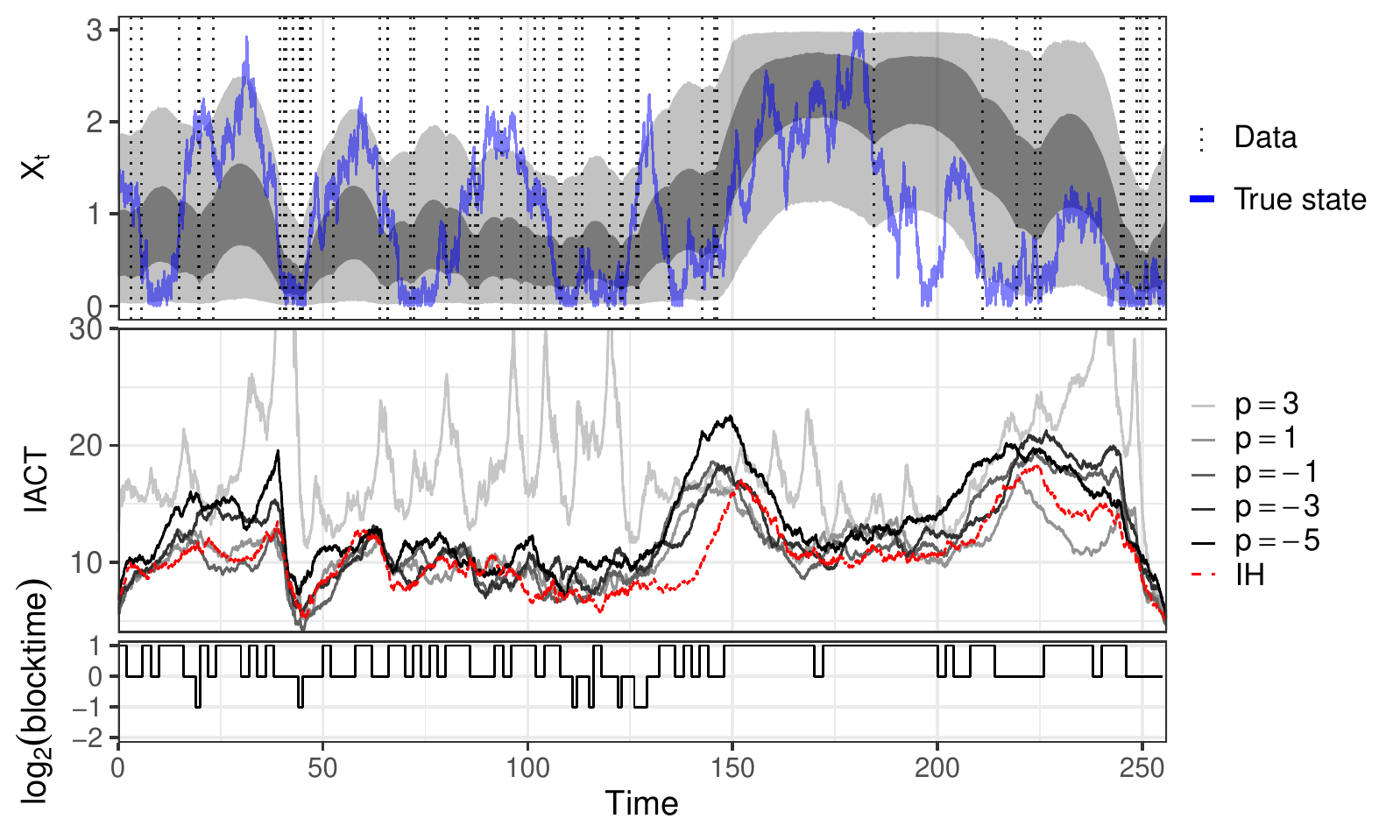}
  \caption{\small Top: the observations $\tilde{\tau}$,  the true simulated state and 
  the 50\% and 95\% credible intervals for $X_t \mid \tilde{\tau}$ (shaded)
    of the CP-RBM model. Middle: the IACTs with homogeneous blocking with blocktime $2^p$  (shades of gray), and with the inhomogeneous blocking $\mathsf{IH}$ (red dashed line). Bottom: the inhomogeneous blocking from Algorithm \ref{alg:choose-blocking} using $n = 50$ and $N = 8$.}
  \label{fig:cp-rbm-summary}
\end{figure}

In terms of the IACT, the blocking sequence returned by Algorithm \ref{alg:choose-blocking} appears to perform
similarly to the best choices for the blocking sequences constructed with constant blocktimes, indicating that
the method here provides adequate performance without trial runs of the CPF-BBS.
The bottom pane shows how the blocktime of the inhomogeneous blocking switches between $1/2, 1$ and $2$.

\subsection{Movement modelling with 
terrain preference} \label{sec:movement-modelling}

We conclude with an application of the CPF-BBS to a movement modelling scenario.
Here, we are interested in modelling the movement of an object on a plane based on noisy observations and knowledge of terrain in which the object moves.
We assume that the object has a `preference' for spending time in certain terrain types. 

To model such a setting, we build on the continuous-time correlated random walk (CTCRW) model suggested for animal movement modelling based on telemetry data \citep{johnson-ctcrw}.
The dynamics of the CTCRW model arise from a special case of the SDE \eqref{eq:ctcrwp-sde}, obtained by setting $\beta_x = 0$ and denoting $\beta \defeq \beta_v$. 
Using this SDE independently in $x$ and $y$ dimensions 
yields a 4-dimensional state $X_t = (V_t^{(x)}, L_t^{(x)}, V_t^{(y)}, L_t^{(y)})^T$ and a movement model on the plane, which we call the CTCRW SDE. 
The full CTCRW model also incorporates two-dimensional location observations $z = (z_k)_{k = 1, 2, \dots, K_z}$ observed at times $(\tilde{t}_k)_{k=1, 2, \ldots, K_z}$.
Each observation is related to the location state variables, $\mathbf{L}_{t} = (L_{t}^{(x)}, L_{t}^{(y)})^{T}$, with 
$z_k = \mathbf{L}_{\tilde{t}_k} + \epsilon_{k}$,
where $\epsilon_k \sim N(0, \eta^2 I_2)$, where $\eta$ is a standard deviation and $I_2$ stands for the $2 \times 2$ identity matrix.
We use the initial distribution $X_{t_1} \sim N( (0, z_{11}, 0, z_{12})^{T}, \mathrm{diag}(\sigma_V^2, \sigma_L^2, \sigma_V^2, \sigma_L^2))$,
where $z_{11}$ and $z_{12}$ are the first and second coordinates of the first observation, respectively, $\sigma_V^{2} = \sigma^2/(2\beta)$ (the stationary variance of the velocity component) and $\sigma_L^2$ is a parameter.
The details regarding the solution of the CTCRW SDE are given in Appendix \appref{sec:ctcrw-details}.

Our model, which we denote CTCRW-T ($T$ standing for `terrain') differs from the CTCRW model of \cite{johnson-ctcrw} by incorporating the effect of terrain. 
We use a discretisation of the CTCRW SDE conditioned on the observations as the sequence of $M_k$'s in the FK representation of the CTCRW-T. 
More specifically, we define 
\begin{equation}
  \label{eq:ctcrwt-M}
  \begin{aligned}
    M_1 & = \mathrm{Law} (X_{t_1} \mid Z = z), \\
    M_k(\uarg\mid x) & = \mathrm{Law} (X_{t_{k}} \mid X_{t_{k-1}} = x, Z = z), \qquad \text{for } 2 \leq k \leq T, 
  \end{aligned}
\end{equation}
where $X_t$ stands for the state of the CTCRW model at time $t$ and $Z$ stands for all observations and $z$ their realised values. 
The distributions in \eqref{eq:ctcrwt-M} are Gaussian, and they can be computed as marginal and conditional distributions of the joint normal distribution \eqref{eq:joint-normal} in Appendix \appref{sec:lgssm-conditionals}. 

The CTCRW-T models terrain preference through its potentials $G_k$ that are of the form
\eqref{eq:fk-path-integral-G}
with $\mathcal{V}(x) = -\log{(v_i)}$ when $x$ is in terrain $i$.
We call the values $v_i \in [0, 1]$, $i = 1, \dots, K_T$, `terrain coefficients', which induce the potential values for each of the $K_T$ terrain types.

We apply the CTCRW-T model in a region of Finland containing lakes, plotted in the background of Figure \ref{fig:ctcrw-vs-ctcrwt-trajectories}. 
The colors of the background map depict the value of $\mathcal{V}$, with black representing larger values, that is, lower potential. 
We define the terrain types based on the Corine Land Cover classification \citep{corine} which classifies each $20\times 20$ metre cell in Finland to one of five classes.  
The terrain types and their associated terrain coefficients (in parentheses) are `Artificial surfaces' (0.2), `Agricultural areas' (0.6), `Forests and semi-natural areas' (0.5), `Water bodies' (0.0) and `Wetlands' (0.5).
The terrain coefficient of `Water bodies' is set to zero, since we want to constrain the movement on land only.  

With the potential map constructed this way, we set $\tau = 16$ and hand-picked 16 observed locations in a clockwise pattern around the lakes, spacing the observation times equidistantly in time.
The observed locations appear as crosses in Figure \ref{fig:ctcrw-vs-ctcrwt-trajectories}.
\begin{figure}
  \centering
  \includegraphics{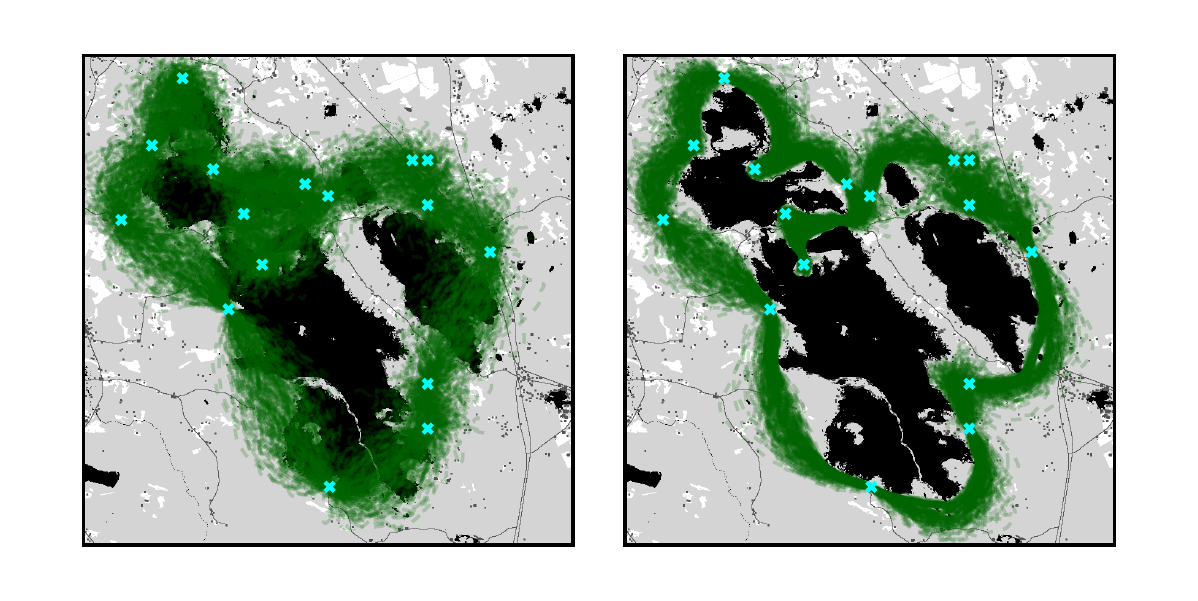}
  \caption{\small Comparison of CTCRW (left) and CTCRW-T (right) with 250 simulated location trajectories.
  The observations used by both models are shown with crosses. The CTCRW-T also uses terrain, which includes lakes (black). 
}
  \label{fig:ctcrw-vs-ctcrwt-trajectories}
\end{figure}
The CTCRW model parameters $\beta$ and $\sigma$ were fit via maximum likelihood, and we set $\eta = 50$ and $\sigma_L = 50$.

We then applied the CPF-BBS with systematic resampling, $N = 16$ and $\blocktime = 1.0$ for $11000$ iterations, discarding the first $1000$ as burn-in. 
$|\Delta_k|$ was set to $2^{-7}$.
The right pane of Figure \ref{fig:ctcrw-vs-ctcrwt-trajectories} shows 250 of the simulated location trajectories from the CTCRW-T model. 
In comparison, the left pane shows trajectories simulated from the CTCRW model conditioned on the observed locations, simulated using \eqref{eq:ctcrwt-M}.
We observe that the trajectories simulated from the CTCRW-T model are influenced by the conditioning on the observations, while avoiding water bodies, as desired. 

We also tested the performance of the CPF-BBS with the blocking sequence obtained using Algorithm \ref{alg:choose-blocking} (using $N = 512$ and $n = 25$), as well as CPF-BS in this example. 
Here, the number of particles for Algorithm \ref{alg:choose-blocking} had to be set slightly higher to ensure that a sufficient number of particles end up in regions of positive potential (due to the hard constraint induced by `Water bodies').
Figure \ref{fig:terrain-sim-bbcpf-vs-cpfbs-iacts} compares the three algorithms by plotting the IACT of the state variable $L_{t}^{(x)}$ with respect to time.
\begin{figure}
  \centering
  \includegraphics{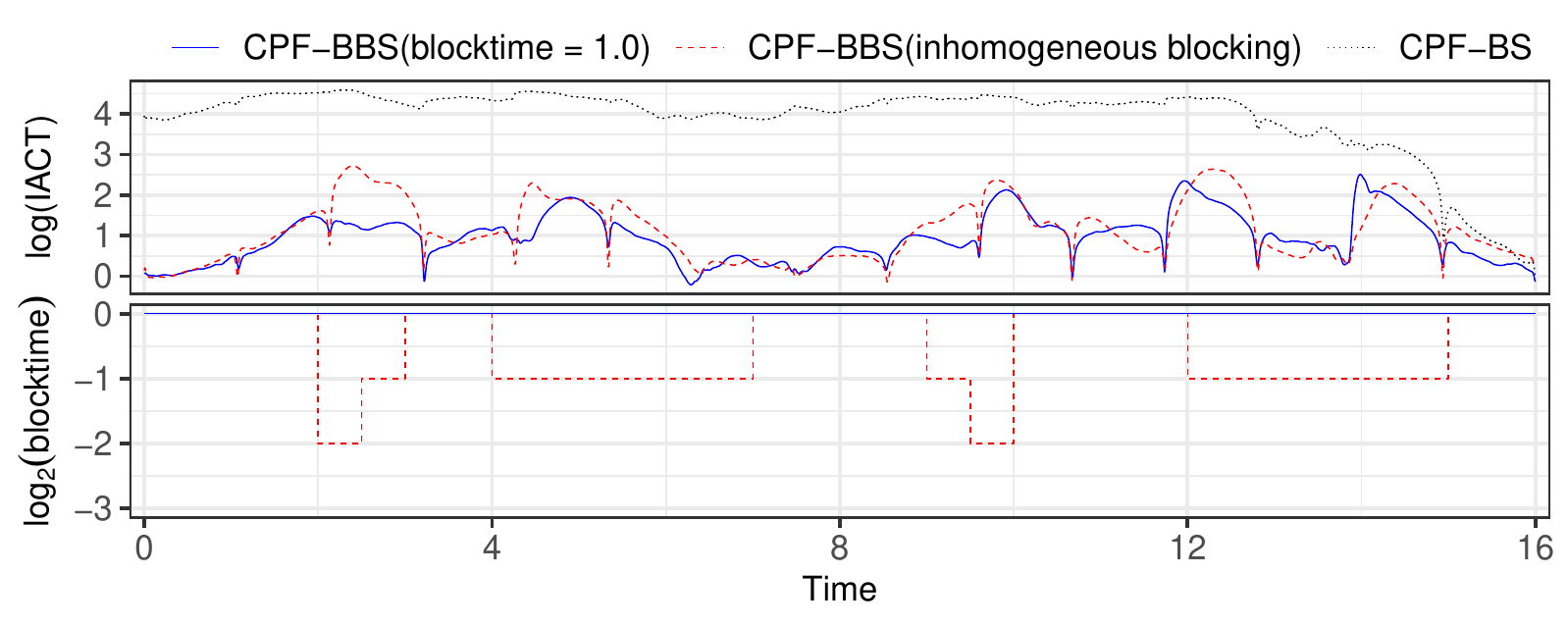}
  \caption{\small The logarithmic IACT (top) of the state $L_t^{(x)}$ in the CTCRW-T model
  with $|\Delta_k| = 2^{-7}$ for CPF-BS and the CPF-BBS with blocktime $1.0$
  and the blocking from Algorithm \ref{alg:choose-blocking} (bottom).}
  \label{fig:terrain-sim-bbcpf-vs-cpfbs-iacts}
\end{figure}
The plots for the other state variables were similar.
Clearly, the simulation efficiencies of both variants of the CPF-BBS are superior here in comparison to the CPF-BS.
Between the automatic blocking and constant blocking, the finding is similar as with the CP-RBM model: 
the blocking optimisation via Algorithm \ref{alg:choose-blocking} yields similar results as the `hand tuned' constant blocking with $\blocktime = 1.0$.
The supplementary material also includes an animation that visualises the values of all sampled trajectories at each time point of the simulation, 
showing slower exploration of the target distribution using the CPF-BS.

We also experimented with the above three algorithms using a higher value for $|\Delta_k|$, a situation where a greater discretisation error in the approximation \eqref{eq:path-integral-approx} may be tolerated.
We found that when $|\Delta_k|$ was increased to 0.125, the resulting IACTs of $L_t^{(x)}$ were similar between the three algorithms (see Figure \sref{fig:terrain-sim-bbcpf-vs-cpfbs-iacts-supp} in the supplementary material).

\section{Discussion} \label{sec:discussion} 

The methods presented in this paper make inference more efficient (and feasible) for an important class of statistical models, which includes hidden Markov models (HMMs) involving weakly informative observations and slowly mixing dynamics and, in particular, time-discretisations of continuous-time path integral models.

Our first contribution was presenting two new conditional resampling algorithms for CPFs in such a context: the killing resampling $\rho_{\mathrm{kill}}$, and the systematic resampling with mean partitioned weights $\rho_{\mathrm{syst}}$. Based on our experiments, $\rho_{\mathrm{syst}}$ performs slightly better than $\rho_{\mathrm{kill}}$, coinciding with the recent theoretical and empirical findings of \cite{chopin-singh-soto-vihola} for the particle filter in a similar context. 
Based on our findings, we generally recommend to use $\rho_{\mathrm{syst}}$ with the CPF in the weakly informative regime, but the simpler killing resampling $\rho_{\mathrm{kill}}$ can also be sufficient for most purposes.

Adaptive resampling \citep{liu-chen-blind} can also be used to reduce resamplings and thereby make the particle filter stable in the weakly informative regime. Adaptive resampling has been suggested also in the context of conditional particle filter \citep[Algorithm 5.3]{lee-phd}. However, it is not obvious how to implement a valid (bridge) backward sampling with adaptive resampling.

Our main contribution is a new CPF with bridge backward sampling (CPF-BBS), which may be regarded as a generalisation of the celebrated CPF with backward sampling (CPF-BS) \citep{whiteley-backwards-note}.
The key ingredient of the CPF-BBS which avoids performance issues of the CPF-BS in the weak potentials and slowly mixing context, is the bridging CPF step that updates the latent trajectory subject to a blocking sequence that acts as a tuning parameter of the method.
We presented a computationally cheap procedure for finding an appropriate blocking sequence, which
is based on a proxy of the integrated autocorrelation time of the output Markov chain, the so-called
probability of lower boundary updates ($\plu$), which measures the probability that the bridge CPF updates the value at the block lower boundary.
We derived an estimator for $\plu$ that we suggest to use for blocking sequence tuning via Algorithm \ref{alg:choose-blocking} that uses a small number of trial runs of the standard particle filter with ancestor tracing to estimate $\plu$ prior to running the CPF-BBS.

The CPF-BBS is generally applicable, assuming that the conditional distributions $M_{u \mid \ell}$ and $M_{k \mid k-1, u}$ related to the individual blocks $(\ell, u)$ may be computed (Assumption \ref{a:conditionals}). In principle, it is always possible to choose such proposals $M_k$, but 
careless choice might result in informative potentials $G_k$ and therefore poor performance. 
The contrary is also possible: with suitably chosen $M_k$, the $G_k$ can be weakly informative, even if the HMM observations are informative. This can be achieved by designing $M_k$ and $G_k$ by suitable `lookaheads' or `twisting' \citep{guarniero-johansen-lee}, such as an approximate smoothing distribution from a Laplace approximation \citep[cf.][Section 8.1]{vihola-helske-franks}.

The experiments suggest that our estimator for $\plu$ is in good agreement with the true $\plu$. 
Algorithm \ref{alg:choose-blocking}, which finds an appropriate blocking automatically, showed promising behaviour in our experiments, leading to performance similar to `hand tuning' the blocking sequence.
Using Algorithm \ref{alg:choose-blocking} is easy: it only requires the user to specify the number of iterations and number particles used in the selection to obtain adequate performance `out of the box'.
In all of the examples we studied, we found $50$ iterations to suffice for block selection, but we presume that the number of particles has to be chosen in a model by model basis.

The performance of the CPF-BBS in practice was promising: we found that the method can provide a substantial performance improvement over CPF-BS in the weak potentials and slowly mixing dynamics setting. This was particularly clear with our movement modelling experiment, which can be of independent interest in certain applications. For instance, our model could be a potentially useful alternative modelling approach for `step-selection' type analyses for territorial animals.

We believe that $\plu$ and the ideas in the estimator we derived for it can be of interest in other contexts, too.
In Section \ref{sec:plu-estimator} we discussed the possibility of obtaining estimates for $\plu$ for $N \neq N_0$, where $N_0$ is the number of particles used for the necessary computations. We found empirically (results not reported) that the agreement between $\plu$ and $\widehat{\plu}$ remains similar as in Figure \ref{fig:blocksize-tuning-ctcrwp} if we use this alternative estimation procedure. 
This method could potentially be elaborated to a heuristic for choosing the number of particles $N$ for the CPF-BBS. One potential way forward is to determine a `cutoff level' for how large a $\plu$ is `large enough,' and the smallest $N$ reaching this level would be chosen. 
Further developments of these ideas are out of the scope of the present paper.

In some applications relevant for the weakly informative context, the initial distribution $M_1$ can be diffuse (relative to the smoothing distribution) --- even an (improper) uniform measure. In such a case, the CPF and also the CPF-BBS  will suffer from poor mixing, but there are relatively direct extensions that are applicable also with the CPF-BBS. Indeed, \cite{fearnhead-melgkotsidou} discuss general state augmentations that can be useful, and a straightforward implementation is often possible in terms of $M_1$-reversible transitions \citep{karppinen-vihola}.

Multilevel Monte Carlo -type methods \citep[e.g.][]{vihola-unbiased} have been recently used for inference with increasingly refined discretisations of continuous-time models. 
The CPF-BBS could be useful for devising multilevel estimators for path integral models.

\section*{Acknowledgements}

SK and MV were supported by the Academy of Finland grants 315619 and 346311. 
The authors wish to thank Anthony Lee and Nicolas Chopin for useful discussions and acknowledge CSC --- IT Center for Science, Finland, for computational resources.

\section*{Supplementary Materials}

\begin{description}
  \item[Source codes] are available at \url{https://github.com/skarppinen/cpf-bbs}.
\end{description}

\bibliographystyle{abbrvnat}
\bibliography{refs.bib}

\begin{thebibliography}{43}
\providecommand{\natexlab}[1]{#1}
\providecommand{\url}[1]{\texttt{#1}}
\expandafter\ifx\csname urlstyle\endcsname\relax
  \providecommand{\doi}[1]{doi: #1}\else
  \providecommand{\doi}{doi: \begingroup \urlstyle{rm}\Url}\fi

\bibitem[Andrieu et~al.(2010)Andrieu, Doucet, and
  Holenstein]{andrieu-doucet-holenstein}
C.~Andrieu, A.~Doucet, and R.~Holenstein.
\newblock Particle {M}arkov chain {M}onte {C}arlo methods.
\newblock \emph{J. R. Stat. Soc. Ser. B Stat. Methodol.}, 72\penalty0
  (3):\penalty0 269--342, 2010.

\bibitem[Andrieu et~al.(2018)Andrieu, Lee, and Vihola]{andrieu-lee-vihola}
C.~Andrieu, A.~Lee, and M.~Vihola.
\newblock Uniform ergodicity of the iterated conditional {SMC} and geometric
  ergodicity of particle {G}ibbs samplers.
\newblock \emph{Bernoulli}, 24\penalty0 (2):\penalty0 842--872, 2018.

\bibitem[Arnaudon and Del~Moral(2020)]{arnaudon-delmoral}
M.~Arnaudon and P.~Del~Moral.
\newblock A duality formula and a particle {G}ibbs sampler for continuous time
  {F}eynman-{K}ac measures on path spaces.
\newblock \emph{Electron. J. Probab.}, 25:\penalty0 1--54, 2020.

\bibitem[Carter et~al.(2014)Carter, Mendes, and Kohn]{carter-mendes-kohn}
C.~K. Carter, E.~F. Mendes, and R.~Kohn.
\newblock An extended space approach for particle {M}arkov chain {M}onte
  {C}arlo methods.
\newblock Preprint arXiv:1406.5795, 2014.

\bibitem[Chopin and Singh(2015)]{chopin-singh}
N.~Chopin and S.~S. Singh.
\newblock On particle {G}ibbs sampling.
\newblock \emph{Bernoulli}, 21\penalty0 (3):\penalty0 1855--1883, 2015.

\bibitem[Chopin et~al.(2022)Chopin, Singh, Soto, and
  Vihola]{chopin-singh-soto-vihola}
N.~Chopin, S.~S. Singh, T.~Soto, and M.~Vihola.
\newblock On resampling schemes for particle filters with weakly informative
  observations.
\newblock \emph{Ann. Statist.}, 50\penalty0 (6):\penalty0 3197--3222, 2022.

\bibitem[Crisan et~al.(1999)Crisan, Del~Moral, and Lyons]{crisan-discrete}
D.~Crisan, P.~Del~Moral, and T.~Lyons.
\newblock Discrete filtering using branching and interacting particle systems.
\newblock \emph{Markov Process. Related Fields}, 5\penalty0 (3):\penalty0
  293--318, 1999.

\bibitem[d'Avigneau et~al.(2022)d'Avigneau, Singh, and Ober]{dAvigneau-SMM}
A.~M. d'Avigneau, S.~S. Singh, and R.~J. Ober.
\newblock Limits of accuracy for parameter estimation and localization in
  single-molecule microscopy via sequential {M}onte {C}arlo methods.
\newblock \emph{SIAM Journal on Imaging Sciences}, 15\penalty0 (1):\penalty0
  139--171, 2022.

\bibitem[de~Jong and Mackinnon(1988)]{smoothed-covariances}
P.~de~Jong and M.~J. Mackinnon.
\newblock Covariances for smoothed estimates in state space models.
\newblock \emph{Biometrika}, 75\penalty0 (3):\penalty0 601--602, 1988.

\bibitem[Del{~}Moral(2004)]{del-moral}
P.~Del{~}Moral.
\newblock \emph{{F}eynman-{K}ac Formulae}.
\newblock Springer, 2004.

\bibitem[Del~Moral(2013)]{del-moral-2013}
P.~Del~Moral.
\newblock \emph{Mean field simulation for {M}onte {C}arlo integration}.
\newblock Chapman and Hall/CRC, 2013.

\bibitem[Del~Moral and Murray(2015)]{delmoral-murray}
P.~Del~Moral and L.~M. Murray.
\newblock Sequential {M}onte {C}arlo with highly informative observations.
\newblock \emph{SIAM/ASA Journal on Uncertainty Quantification}, 3\penalty0
  (1):\penalty0 969--997, 2015.

\bibitem[Doucet et~al.(2006)Doucet, Briers, and
  S{\'e}n{\'e}cal]{doucet-briers-senecal}
A.~Doucet, M.~Briers, and S.~S{\'e}n{\'e}cal.
\newblock Efficient block sampling strategies for sequential {M}onte {C}arlo
  methods.
\newblock \emph{J. Comput. Graph. Statist.}, 15\penalty0 (3):\penalty0
  693--711, 2006.

\bibitem[Durbin and Koopman(2012)]{durbin-koopman}
J.~Durbin and S.~J. Koopman.
\newblock \emph{Time series analysis by state space methods}.
\newblock Oxford University Press, New York, 2nd edition, 2012.

\bibitem[Fearnhead and K{\"u}nsch(2018)]{fearnhead-kunsch}
P.~Fearnhead and H.~R. K{\"u}nsch.
\newblock Particle filters and data assimilation.
\newblock \emph{Annu. Rev. Stat. Appl.}, 5:\penalty0 421--449, 2018.

\bibitem[Fearnhead and Meligkotsidou(2016)]{fearnhead-melgkotsidou}
P.~Fearnhead and L.~Meligkotsidou.
\newblock Augmentation schemes for particle {MCMC}.
\newblock \emph{Statist. Comput.}, 26\penalty0 (6):\penalty0 1293--1306, 2016.

\bibitem[{Finnish Environment Institute SYKE}(2018)]{corine}
{Finnish Environment Institute SYKE}.
\newblock {CORINE Land Cover 2018}.
\newblock The data are downloaded from the Data Download Service of SYKE on
  03.12.2018 under the license CC 4.0 BY, 2018.

\bibitem[Flegal and Jones(2010)]{flegal-batch}
J.~M. Flegal and G.~L. Jones.
\newblock Batch means and spectral variance estimators in {M}arkov chain
  {M}onte {C}arlo.
\newblock \emph{Ann. Statist.}, 38\penalty0 (2):\penalty0 1034--1070, 2010.

\bibitem[Gerber et~al.(2019)Gerber, Chopin, and
  Whiteley]{gerber-chopin-whiteley}
M.~Gerber, N.~Chopin, and N.~Whiteley.
\newblock Negative association, ordering and convergence of resampling methods.
\newblock \emph{Ann. Statist.}, 47\penalty0 (4):\penalty0 2236--2260, 2019.

\bibitem[Glynn and Whitt(1992)]{glynn1992asymptotic}
P.~W. Glynn and W.~Whitt.
\newblock The asymptotic efficiency of simulation estimators.
\newblock \emph{Operations research}, 40\penalty0 (3):\penalty0 505--520, 1992.

\bibitem[Guarniero et~al.(2017)Guarniero, Johansen, and
  Lee]{guarniero-johansen-lee}
P.~Guarniero, A.~M. Johansen, and A.~Lee.
\newblock The iterated auxiliary particle filter.
\newblock \emph{J. Amer. Statist. Assoc.}, 112\penalty0 (520):\penalty0
  1636--1647, 2017.

\bibitem[Hoare(1962)]{hoare-quicksort}
C.~A. Hoare.
\newblock Quicksort.
\newblock \emph{The Computer Journal}, 5\penalty0 (1):\penalty0 10--16, 1962.

\bibitem[Hooten et~al.(2017)Hooten, Johnson, McClintock, and
  Morales]{hooten-johnson-mcclintock-morales}
M.~B. Hooten, D.~S. Johnson, B.~T. McClintock, and J.~M. Morales.
\newblock \emph{Animal Movement: Statistical Models for Telemetry Data}.
\newblock CRC Press, 2017.
\newblock ISBN 978-1-4665-8214-9.

\bibitem[Johnson et~al.(2008)Johnson, London, Lea, and Durban]{johnson-ctcrw}
D.~S. Johnson, J.~M. London, M.-A. Lea, and J.~W. Durban.
\newblock Continuous-time correlated random walk model for animal telemetry
  data.
\newblock \emph{Ecology}, 89\penalty0 (5):\penalty0 1208--1215, 2008.

\bibitem[Karppinen and Vihola(2021)]{karppinen-vihola}
S.~Karppinen and M.~Vihola.
\newblock Conditional particle filters with diffuse initial distributions.
\newblock \emph{Statist. Comput.}, 31\penalty0 (3):\penalty0 1--14, 2021.

\bibitem[Lee(2011)]{lee-phd}
A.~Lee.
\newblock \emph{On auxiliary variables and many-core architectures in
  computational statistics}.
\newblock PhD thesis, 2011.

\bibitem[Lee et~al.(2020)Lee, Singh, and Vihola]{lee-singh-vihola}
A.~Lee, S.~S. Singh, and M.~Vihola.
\newblock Coupled conditional backward sampling particle filter.
\newblock \emph{Ann. Statist.}, 48\penalty0 (5):\penalty0 3066--3089, 2020.

\bibitem[Lindsten et~al.(2014)Lindsten, Jordan, and
  Sch{\"o}n]{lindsten-jordan-schon}
F.~Lindsten, M.~I. Jordan, and T.~B. Sch{\"o}n.
\newblock Particle {G}ibbs with ancestor sampling.
\newblock \emph{J. Mach. Learn. Res.}, 15\penalty0 (1):\penalty0 2145--2184,
  2014.

\bibitem[Lindsten et~al.(2015{\natexlab{a}})Lindsten, Bunch, Singh, and
  Sch{\"o}n]{lindsten-bunch-singh-schon}
F.~Lindsten, P.~Bunch, S.~S. Singh, and T.~B. Sch{\"o}n.
\newblock Particle ancestor sampling for near-degenerate or intractable state
  transition models.
\newblock Preprint arXiv:1505.06356, 2015{\natexlab{a}}.

\bibitem[Lindsten et~al.(2015{\natexlab{b}})Lindsten, Douc, and
  Moulines]{lindsten-douc-moulines}
F.~Lindsten, R.~Douc, and E.~Moulines.
\newblock Uniform ergodicity of the particle {G}ibbs sampler.
\newblock \emph{Scand. J. Stat.}, 42\penalty0 (3):\penalty0 775--797,
  2015{\natexlab{b}}.

\bibitem[Liu and Chen(1995)]{liu-chen-blind}
J.~S. Liu and R.~Chen.
\newblock Blind deconvolution via sequential imputations.
\newblock \emph{J. Amer. Statist. Assoc.}, 90\penalty0 (430):\penalty0
  567--576, 1995.

\bibitem[Miasojedow and Niemiro(2015)]{miasojedow-niemiro}
B.~Miasojedow and W.~Niemiro.
\newblock Particle {G}ibbs algorithms for {M}arkov jump processes.
\newblock Preprint arXiv:1505.01434, 2015.

\bibitem[Mider et~al.(2021)Mider, Schauer, and Van~der
  Meulen]{mider-schauer-vandermeulen}
M.~Mider, M.~Schauer, and F.~Van~der Meulen.
\newblock Continuous-discrete smoothing of diffusions.
\newblock \emph{Electron. J. Statist.}, 15\penalty0 (2):\penalty0 4295--4342,
  2021.

\bibitem[Mirauta et~al.(2014)Mirauta, Nicolas, and Richard]{mirauta-genetics}
B.~Mirauta, P.~Nicolas, and H.~Richard.
\newblock Parseq: reconstruction of microbial transcription landscape from
  {RNA}-{S}eq read counts using state-space models.
\newblock \emph{Bioinformatics}, 30\penalty0 (10):\penalty0 1409--1416, 2014.
\newblock \doi{https://doi.org/10.1093/bioinformatics/btu042}.

\bibitem[Murray et~al.(2016)Murray, Lee, and Jacob]{murray-lee-jacob}
L.~M. Murray, A.~Lee, and P.~E. Jacob.
\newblock Parallel resampling in the particle filter.
\newblock \emph{J. Comput. Graph. Statist.}, 25\penalty0 (3):\penalty0
  789--805, 2016.

\bibitem[Rasmussen et~al.(2011)Rasmussen, Ratmann, and
  Koelle]{rasmussen-epidemiology}
D.~A. Rasmussen, O.~Ratmann, and K.~Koelle.
\newblock Inference for nonlinear epidemiological models using genealogies and
  time series.
\newblock \emph{PLOS Computational Biology}, 7:\penalty0 1--11, 2011.
\newblock \doi{https://doi.org/10.1371/journal.pcbi.1002136}.

\bibitem[S{\"a}rkk{\"a} and Solin(2019)]{sarkka-solin}
S.~S{\"a}rkk{\"a} and A.~Solin.
\newblock \emph{Applied stochastic differential equations}, volume~10.
\newblock Cambridge University Press, 2019.

\bibitem[Singh et~al.(2017)Singh, Lindsten, and
  Moulines]{singh-lindsten-moulines}
S.~S. Singh, F.~Lindsten, and E.~Moulines.
\newblock Blocking strategies and stability of particle {G}ibbs samplers.
\newblock \emph{Biometrika}, 104\penalty0 (4):\penalty0 953--969, 2017.

\bibitem[Thurfjell et~al.(2014)Thurfjell, Ciuti, and
  Boyce]{thurfjell-ciuti-boyce}
H.~Thurfjell, S.~Ciuti, and M.~S. Boyce.
\newblock Applications of step-selection functions in ecology and conservation.
\newblock \emph{Movement Ecology}, 2\penalty0 (4), 2014.
\newblock \doi{https://doi.org/10.1186/2051-3933-2-4}.

\bibitem[Vihola(2018)]{vihola-unbiased}
M.~Vihola.
\newblock Unbiased estimators and multilevel {M}onte {C}arlo.
\newblock \emph{Oper. Res.}, 66\penalty0 (2):\penalty0 448--462, 2018.

\bibitem[Vihola et~al.(2020)Vihola, Helske, and Franks]{vihola-helske-franks}
M.~Vihola, J.~Helske, and J.~Franks.
\newblock Importance sampling type estimators based on approximate marginal
  {MCMC}.
\newblock \emph{Scand. J. Stat.}, 47\penalty0 (4):\penalty0 1339--1376, 2020.

\bibitem[Whiteley(2010)]{whiteley-backwards-note}
N.~Whiteley.
\newblock Discussion on {P}article {M}arkov chain {M}onte {C}arlo methods.
\newblock \emph{J. R. Stat. Soc. Ser. B Stat. Methodol.}, 72\penalty0
  (3):\penalty0 306--307, 2010.

\bibitem[Wood(2010)]{wood-ecology}
S.~N. Wood.
\newblock Statistical inference for noisy nonlinear ecological dynamic systems.
\newblock \emph{Nature}, 466:\penalty0 1102--1104, 2010.
\newblock \doi{https://doi.org/10.1038/nature09319}.

\end{thebibliography}

\appendix

\section{Validity of CPF with killing and systematic resampling} 
\label{app:cpf} 

We start by stating an easy lemma, whose proof is immediate.
\begin{lemma}
    \label{lem:valid-conditional-resampling} 
For a valid conditional resampling scheme $r^{(p,n)}(\uarg \mid
g^{(1:N)})$ and its unconditional version $r$, it holds that:
\begin{enumerate}[(i)]
\item $\E_{r(\uarg\mid g^{(1:N)})}\big[\sum_{i=1}^N \charfun{A^{(i)}  = j}\big] = N
    \frac{g^{(j)}}{\sum_{i=1}^N g^{(i)}}$ for all $j\in\{1{:}N\}$, and
\item \label{item:unconditional-vs-conditional}
  $r(a^{(1:N)}\mid g^{(1:N)}) \charfun{\smash{a^{(n)}} = p}
    = \frac{g^{(p)}}{\sum_{\ell=1}^N g^{(\ell)}}
      r^{(p,n)}(a^{(1:N)}\mid g^{(1:N)})$ for all
      $a^{(i)}$, 
      $n$ and $p$ in $\{1{:}N\}$.
\end{enumerate}
\end{lemma}

In what follows, we denote 
$\underline{M}_1(\underline{x}_1) = \prod_{i=1}^N M_1(x_1^{(i)})$,
$\underline{M}_k(\underline{x}_k\mid x_{k-1}^{(\underline{a})}) = \prod_{i=1}^N
M_k(x_k^{(i)}\mid x_{k-1}^{(a^{(i)})})$ and $G_k(\underline{\vec{x}}_k) = \big(G_k(\vec{x}_k^{(1)}), \ldots, G_k(\vec{x}_k^{(N)})\big)$.

\begin{proof}[Proof of Theorem \mref{thm:cpf-valid}] 
  Assume that $X_{1:T}^* \sim \pi$ and $B_{1:T}\sim \mathcal{U}(\{1{:}N\}^T)$
  independently. The joint distribution of $B_{1:T}$, the particles $\underline{X}_{1:T}$, 
  and the ancestories $\underline{A}_{1:T-1}$ generated by the CPF, 
  may be written as
  \begin{align}
  &\frac{\underline{M}_1(\underline{x}_1)}{\mathcal{Z} N^T}
      \bigg[\prod_{k=2}^T 
        r^{(b_{k-1},b_k)}\big(\underline{a}_{k-1}\mid
        G_{k-1}(\underline{\mathbf{x}}_{k-1})\big)
    \underline{M}_k(\underline{x}_k\mid x_{k-1}^{(\underline{a}_{k-1})})
      G_{k-1}(\mathbf{x}_{k-1}^{(b_{k-1})}) \bigg]
      G_T(\mathbf{x}_T^{(b_T)}) \nonumber\\
      &= \frac{\underline{M}_1(x_1^{(1:N)}) }{\mathcal{Z}}
      \bigg(\prod_{k=1}^{T} \frac{1}{N} \sum_{\ell=1}^N
      G_k(\mathbf{x}_k^{(\ell)})\bigg) 
      \nonumber\\
      &\phantom{=}
      \bigg(\prod_{k=2}^T 
      \charfun{a_{k-1}^{(b_k)}=b_{k-1}}
      r(\underline{a}_{k-1}\mid G_{k-1}(\underline{\mathbf{x}}_{k-1})\big)
      \underline{M}_k(\underline{x}_k\mid x_{k-1}^{(\underline{a}_{k-1})})
      \bigg) \frac{G_T(\mathbf{x}_T^{(b_T)})}{\sum_{\ell=1}^N
        G_T(\mathbf{x}_T^{(\ell)})},
      \label{eq:joint-cpf}
  \end{align}
  by Lemma \ref{lem:valid-conditional-resampling}.
  Including the variables $\tilde{B}_{1:T}$ to \eqref{eq:joint-cpf} adds the following factor:
\begin{equation}
    \bigg[\prod_{k=2}^T 
    \charfun{a_{k-1}^{(\tilde{b}_k)}=\tilde{b}_{k-1}}\bigg]
    \frac{G_T(\mathbf{x}_T^{(\tilde{b}_T)})}{\sum_{\ell=1}^N
      G_T(\mathbf{x}_T^{(\ell)})}
    \label{eq:cpf-new-ref}
\end{equation}
The joint distribution---product of \eqref{eq:joint-cpf}
and \eqref{eq:cpf-new-ref}---is clearly symmetric with respect to
$(b_{1:T},x_{1:T}^{(b_{1:T})})$ and
$(\tilde{b}_{1:T}, x_{1:T}^{(\tilde{b}_{1:T})})$.
\end{proof} 

In what follows, we use the shorthand $\sigma_s \defeq \sigma_s^N$.

\begin{proof}[Proof of Lemma \mref{lem:conditional-killing} \eqref{item:killing-valid}] 
Suppose that $\bar{A}^{(1:N)}\sim \rho(\uarg \mid g^{(1:N)})$, where
$\rho$ is given in \meqref{eq:killing-resampling}.
We first observe that $\rho$ is unbiased: 
\begin{equation}
    \E\left[\sum_{i=1}^{N}\charfun{
      \bar{A}^{(i)}=j} \right] 
    =\frac{g^{(j)}}{g^{\ast}}
    +\sum_{i=1}^{N}\left(1-\frac{g^{(i)}}{g^{\ast}}\right)
    \frac{g^{(j)}}{\sum_{\ell=1}^{N}g^{(\ell)}}
     =N\frac{g^{(j)}}{\sum_{\ell=1}^{N}g^{(\ell)}}.
 \label{eq:mean_rho_bar}
 \end{equation}
Let $S \in \{1{:}N\}$ be an independent uniformly distributed random
variable, and consider $A^{(1:N)} = \bar{A}^{(\sigma_S(1:N))}$.
Then
$A^{(1:N)} \sim \hat{\rho}(\uarg\mid g^{(1:N)})$,
where
\begin{equation}
    \hat{\rho}(a^{(1:N)}\mid g^{(1:N)}) \defeq \frac{1}{N} \sum_{s=1}^N
    \rho(a^{(\sigma_s(1:N))}\mid g^{(1:N)}),
    \label{eq:symmetrised-killing}
\end{equation}
which also clearly unbiased, and 
from \eqref{eq:mean_rho_bar}, it follows that 
\begin{equation}
     \P(A^{(k)} = j) 
     = \frac{1}{N}\E\Big[\sum_{i=1}^N \charfun{\bar{A}^{(i)} = j}\Big]
     = \frac{g^{(j)}}{\sum_{\ell=1}^{N}g^{(\ell)}}.
     \label{eq:killing-marginal-A1}
 \end{equation}
Next we derive the conditional distribution of
$A^{(-k)}$ given $A^{(k)}=i$.
First, because $A^{(k)}=\bar{A}^{(\sigma_S(k))}$, we have
\[
    \P\left(\sigma_S(k) =j\mid A^{(k)}=i\right) 
    =\frac{\P(\sigma_S(k)=j)\P(\bar{A}^{(j)}=i)}{
      \sum_{\ell=1}^{N}\P\left(\sigma_S(k)=\ell\right)\P\left(\bar{A}^{(\ell)}=i\right)}
 =\frac{\sum_{\ell=1}^{N}g^{(\ell)}}{Ng^{(i)}}
 \P\left(\bar{A}^{(j)}=i\right),
 \label{eq:cond_perm}
 \]
 and $\P(\bar{A}^{(j)}=i) = \frac{g^{(i)}}{g^*}\charfun{j=i}
  + \big(1 - \frac{g^{(j)}}{g^*}\big)\frac{g^{(i)}}{\sum_{\ell=1}^N
    g^{(\ell)}}$, so 
a simple calculation yields
 \begin{equation}
     \P(\sigma_S(k)=j \mid A^{(k)} = i) =
     h(j\mid i),\quad\text{where}\quad
     h(j\mid i) \defeq \begin{cases}
\frac{1}{N}\Big( 1 + \frac{\sum_{\ell\neq i} g^{(\ell)}}{g^*}\Big),
  & j = i \\
\frac{1}{N}\Big( 1 - \frac{g^{(j)}}{g^*}\Big), & j\neq i . 
\end{cases}
\label{eq:permutation-distribution}
 \end{equation}
Note that $\sigma_S(k)=j$ is equivalent with $S = \sigma_{-k}(j)$.

We conclude that $A^{(1:N)} \sim \hat{\rho}(\uarg\mid g^{(1:N)})$  may be drawn by first
drawing $B$ from the marginal distribution
of $A^{(k)}$,  that is, $\P(B=i) =
g^{(i)}/\sum_{\ell=1}^N g^{(\ell)}$, drawing $J\sim
h(\uarg\mid B)$, 
setting $S = \sigma_{-k}(J)$
and $\bar{A}^{(\sigma_S(k))}=B$ and
$A^{(j)} = \bar{A}^{(\sigma_S(j))}$ for $i\in\{1{:}N\}$.
\end{proof}

\begin{lemma}
  \label{lem:systematic-invariance}
  Suppose that $\varpi$ is a permutation of $[N]$, and $\varpi_*$ is a cyclic shift of $\varpi$, that is, $\varpi_*(i) = \varpi(\sigma_s(i))$ for some $s \in [N]$, and that 
  \begin{align*}
    \bar{A}^{1:N} &= \varpi(F_\varpi^{-1}(U^{1:N})) \\
    \bar{A}_*^{1:N} &= \varpi_*(F_{\varpi_*}^{-1}(U^{1:N})),
  \end{align*}
where $U^j = \dfrac{j - 1 + U}{N}$ with $U \sim \mathcal{U}(0,1)$.

Then, it holds that $A^{1:N}$ and $A_*^{1:N}$ have the same distribution, where
\begin{align*}
  A^j &= \bar{A}^{\sigma_C(j)} \\
  A_*^{j} &= \bar{A}_*^{\sigma_C(j)},
\end{align*}
and $C \sim \mathcal{U}([N])$ is a random shift offset.
\end{lemma}
\begin{proof}
Without loss of generality, we may consider the case $s=1$ and $\varpi(i) = i$, in which case $\varpi_*(i) = \sigma_1(i)$.

Define $\tilde{U}^i \defeq (U^i - w^1 \mod 1)$, let $j = \arg\min_i \tilde{U}^i$, and let $\tilde{U}^i_* = \tilde{U}^{\sigma_{j}(i)}$.
Observe that $\tilde{U}_*^{1:N}$ and $U^{1:N}$ have the same distribution,
  so the claim follows once we show that $\bar{A}^{1:N} = F^{-1}(U^{1:N})$ and $\bar{A}_*^{1:N} = \sigma_1(F_{\sigma_1}^{-1}(\tilde{U}^{1:N}_*))$ are equal, up to a cyclic shift. Indeed, we will see that for all $i\in[N]$ and $0\le k\le N-1$:
$$
\bar{A}_*^{\sigma_{-j}(i)} = \sigma_1(F_{\sigma_1}^{-1}(\tilde{U}^i))
= k+1 \iff \bar{A}^{i} = k+1.
$$
Let us first assume $k\ge 1$, then the expression on the left is equivalent to
$$
F_{\sigma_1}(k-1) < \tilde{U}^i \le F_{\sigma_1}(k) \iff F(k)  < \tilde{U}^i + w^1 \le F(k+1),
$$
  because $F_{\sigma_1}(\ell) = F(\ell+1) - w^1$. Whenever $U^i - w^1 \ge 0$, we have $\tilde{U}^i + w^1 = U^i$, and the expression on the right simplifies to $\bar{A}^i = F^{-1}(U^i) = k + 1$, as desired.

Suppose then that $U^i - w^1 < 0$, in which case $\bar{A}^i = F^{-1}(U^i) = 1$. But then also $\tilde{U}^i = U^i - w^1 + 1 \in (1-w^1, 1)$, which is equivalent to $F_{\sigma_1}^{-1}(\tilde{U}^i) = N$.
\end{proof}

\begin{proof}[Proof of Lemma \mref{lem:conditional-killing} \eqref{item:systematic-valid}]
  Assume that $\varpi$ is a permutation (such as the mean partition order).
  Let $I^{1:N} = F_{\tilde{\varpi}}^{-1}(U^{1:N})$, with
  \[
    U^{i} = \dfrac{i - 1 + U}{N},
  \]
  with $U \sim \mathcal{U}(0,1)$, that is, standard systematic resampling (Definition \mref{def:sys-res}) with weights $W_{\tilde{\varpi}}^{1:N}$, where $W_{\tilde{\varpi}}^j = W^{\tilde{\varpi}(j)}$ and $\tilde{\varpi}(j) = \varpi(\sigma_{s-1}(j))$, with $s = \varpi^{-1}(i)$.

  Hence, $\tilde{\varpi}$ satisfies
  \begin{equation}
    \tilde{\varpi}(1) = \varpi(\sigma_{s - 1}(1)) = \varpi(s) = i.
  \end{equation}
  Define $\bar{A}^{1:N}$ such that
  \[
    \bar{A}^j = \tilde{\varpi}(I^j).
  \] 
  Then, by Lemma \ref{lem:systematic-invariance}, it holds that 
  \[
    A^j = \bar{A}^{\sigma_C(j)}, \text{ for } j \in [N] \text{, with } C \sim \mathcal{U}([N]),
  \]
  have the same distribution as the indices from systematic resampling with order $\varpi$ 
  that have been shifted by $\sigma_C$.
  In particular, note that Definition \mref{def:valid-conditional-resampling} (\ref{defcond:symmetricity}) holds for the latter.

  Consider then the count of indices equal to $i$:
  \[
    N^i = \#\{j: A^j = i\}.
  \]
  Since $A^j = i \iff I^{\sigma_C(j)} = 1$ and the indices $I^{1:N}$ are ascending, it holds that
  \[
    N^i = \max\{j \geq 1 \given I^j = 1\},
  \]
  where $\max$ is zero in case the set is empty.
  The event $N^i = n$ is equivalent with
  \begin{align*}
    \phantom{\iff}& \dfrac{n - 1 + U}{N} < F_{\tilde{\varpi}}(1) \leq \dfrac{n + U}{N} \\
    \iff& n - 1 + U < Nw^i \leq n + U \\
    \iff& Nw^i - (n - 1) > U \geq Nw^i - n.
  \end{align*}
  We deduce that only two values of $n$ have nonzero probability (for $U\in(0,1)$), since:
  \begin{align*}
    n &= \lfloor N w^i \rfloor \iff U \in [r,1) \\
    n &= \lfloor N w^i \rfloor + 1 \iff U \in (0,r),
  \end{align*}
  where $r = Nw^i - \lfloor Nw^i \rfloor$.
  Furthermore, the conditional probabilities for the events $N^i = n$ are given as:
  \[ 
    \mathbb{P}(N^i = n\mid A^k=i)
    = \frac{\mathbb{P}(N^i = n, A^k=i)}{\mathbb{P}(A^k=i)}
    = \frac{\mathbb{P}(N^i = n, A^k=i)}{w^i},
  \]
  where the numerator satisfies
  \begin{align*}
    \mathbb{P}(N^i = n, A^k=i)
    &= \sum_{c=1}^N \mathbb{P}(C=c, A^k=i\mid N^i = n)\mathbb{P}(N^i = n) \\
    &= \sum_{c=1}^N \mathbb{P}(A^k=i\mid C=c, N^i = n)
    \mathbb{P}(C=c\mid N^i = n)\mathbb{P}(N^i = n).
  \end{align*}
  
  Since
  \begin{itemize}
    \item $\mathbb{P}(N^i = \lfloor N w^i \rfloor + 1) = r$ and $\mathbb{P}(N^i = \lfloor N w^i \rfloor) = 1-r$ (from above),
    \item $\mathbb{P}(C=c\mid N^i = n) = 1/N$ (because $C$ is independent of $I^{1:N}$ and therefore $N^i$),
    \item $\mathbb{P}(A^k=i\mid C=c, N^i = n)$ are deterministic, either zero or one, and precisely $n$ are one,
  \end{itemize}
  it holds that 
  \[
    \mathbb{P}(N^i = \lfloor N w^i \rfloor + 1\mid A^k=i) = \frac{(\lfloor N w^i \rfloor + 1)r}{N w^i} := p,
  \]
  and 
  \[
    \mathbb{P}(N^i = \lfloor N w^i \rfloor \mid A^k=i) = 1 - p.
  \]
  Observe also that the random variable $U$ conditional on $A^k=i$ and $N^i = n$ has the density $\mathcal{U}(0,r)$ if $N^i = \lfloor N W^i \rfloor + 1$ and $\mathcal{U}(r,1)$ if $n = \lfloor N W^i \rfloor$.
  This follows since $U$ is conditionally independent from the event $A^k = i$ given $N^i = n$, since $U$ only depends on $A^k = i$ through $N^i$.
  Similarly, 
  \[
    \mathbb{P}(C=c\mid A^k=i, N^i = n, U)
    = \mathbb{P}(C=c\mid A^k=i, N^i = n) = \frac{1}{n} 1(\sigma_c(k)\in[1,n]).
  \]
  In practice, we can simulate $C$ from this distribution as follows:
  \begin{enumerate}
    \item Draw $\bar{C}\sim U\{1,\ldots,n\}$,
    \item Set $C = \sigma_{-k}^N(\bar{C})$,
  \end{enumerate}
since then $\sigma_C(k) = \bar{C}$.

Algorithm \mref{alg:systematic} proceeds by first drawing $n \sim N^i \mid A^k = i$. 
Then, $c \sim C \mid N^i = n, A^k = i$ is drawn, after which $A^{1:N}$ (satisfying $A^k = i$) may be drawn by drawing $\bar{A}^{1:N}$ and shifting by $c$. 

\end{proof}
\section{Validity of CPF-BBS} 
\label{app:bbcpf} 

We start by two auxiliary results about marginal distributions after
partial ancestor tracing and a partial CPF.
In what follows, we assume that $G_k(x_{1:k}) = G_k(x_{k-1:k})$ for
$k\in\{2{:}T\}$.
Using the definition of $\underline{M}_k$ as in Appendix \ref{app:cpf},
let us fix some notation: for $u=1,\ldots,T$, 
denote by 
$\check{\pi}^{(N)}_{u}(\underline{x}_{1:u}, \underline{a}_{1:u-1},
b_u)$:
\begin{align*}
\frac{1}{\mathcal{Z}} \underline{M}_1(\underline{x}_1)
\prod_{k=1}^{u-1} \bigg[
  \bigg( \frac{1}{N} \sum_{j=1}^N
G_k(\vec{x}_{k}^{(j)})\bigg)
 r\big(\underline{a}_{k}\mid
G_{k}(\underline{\vec{x}}_{k} )\big) 
  \underline{M}_{k+1}(\underline{x}_{k+1}\mid
x_{k}^{(\underline{a}_{k})}) \bigg]
\frac{G_u(\vec{x}_u^{(b_u)})}{N},
\end{align*}
and $ \gamma_{u:T}(x_{u:T}) \defeq \prod_{k=u+1}^T M_k(x_k\mid x_{k-1})
    G_k(x_{k-1:k})$ with $\gamma_{T:T}(x_T)\equiv 1$, 
then the following define probability distributions for $u=1,\ldots,T$:
\[
    \mu_{u}^{(N)}(\underline{x}_{1:u}, \underline{a}_{1:u-1}, b_u,
    x_{u+1:T}^*)
    \defeq 
    \check{\pi}^{(N)}_{u}(\underline{x}_{1:u}, \underline{a}_{1:u-1}, b_u)  
  \gamma_{u:T}(x_{u}^{(b_u)}, x_{u+1:T}^*).
\]
\begin{lemma} 
    \label{lem:bbcpf-recurse} 
Suppose that $r^{(p,n)}$ is a valid conditional resampling
scheme, with respect to resampling $r$ in Definition
\mref{def:valid-conditional-resampling}.
Suppose
$(\underline{X}_{1:u},\underline{A}_{1:u-1}, B_u, X_{u+1:T}^*)
\sim \mu_{u}^{(N)}$, and $\ell\in\{1{:}u-1\}$. 
\begin{enumerate}[(i)]
\item \label{item:partial-trace} If $B_{\ell:u-1} \gets
  \textsc{AncestorTrace}(\underline{A}_{\ell:u-1}, B_u)$,
  then the marginal density of
  $\underline{X}_{1:\ell}, \underline{A}_{1:\ell-1}$,
  $X_{\ell+1:u}^{(B_{\ell+1:u})}$, $B_{\ell:u}$ and $X_{u+1:T}^*$ is
  \[
      \mu_{\ell}^{(N)}(\underline{x}_{1:\ell},
  \underline{a}_{1:\ell-1}, b_\ell,
    x_{\ell+1:u}^{(b_{\ell+1:u})}, x_{u+1:T}^*)/N^{u-\ell}.
  \]
\item \label{item:partial-bridge} 
  If further $(X_{\ell:u-1}^*, \tilde{B}_{\ell:u-1})\gets
  \textsc{BridgeCPF}(\underline{\vec{X}}_\ell, B_{\ell:u-1},
  X_{\ell:u}^{B_{\ell:u}})$, $\tilde{B}_u=B_u$ and
  $X_u^*=X_u^{(B_u)}$,
  then the marginal density of 
    $\underline{X}_{1:\ell}, \underline{A}_{1:\ell-1}$,
  $X_{\ell+1:u}^*$ and  
  $\tilde{B}_{\ell:u}$ is
  \[
      \mu_{\ell}^{(N)}(\underline{x}_{1:\ell},
  \underline{a}_{1:\ell-1}, \tilde{b}_\ell, 
    x_{\ell+1:T}^*)/N^{u-\ell}.
  \]
\end{enumerate}
\end{lemma} 
\begin{proof} 
In the case \eqref{item:partial-trace},
the joint density of all variables may be written as
\begin{align*}
&    \check{\pi}_{u}^{(N)}(\underline{x}_{1:u},\underline{a}_{1:u-1}, b_u)
    \bigg(\prod_{k=\ell}^{{u-1}} \charfun{b_k =
      a_k^{(b_{k+1})}}\bigg) \gamma_{u:T}(x_u^{(b_u)}, x_{u+1:T}^*)\\
&   =
   \frac{\check{\pi}_{\ell}^{(N)}(\underline{x}_{1:\ell},\underline{a}_{1:\ell-1},b_\ell)
     }{G_\ell(\vec{x}_\ell^{(b_\ell)})}
       \bigg[\prod_{k=\ell}^{{u-1}} 
       \bigg(  \frac{1}{N}\sum_{i=1}^N G_{k}(\vec{x}_{k}^{(i)})\bigg)
       \charfun{b_k =
      a_k^{(b_{k+1})}}
    r(\underline{a}_{k}\mid G_{k}(\underline{\vec{x}}_{k})\big) \\
& \phantom{
  =
   \frac{\check{\pi}_{\ell}^{(N)}(\underline{x}_{1:\ell},\underline{a}_{1:\ell-1},b_\ell)
     }{G_\ell(\vec{x}_\ell^{(b_\ell)})}
       \bigg[
}
    \underline{M}_{k+1}(\underline{x}_{k+1}\mid
    x_{k}^{(\underline{a}_{k})})
    \bigg] 
    G_u(\vec{x}_u^{(b_u)}) \gamma_{u:T}(x_u^{(b_u)}, x_{u+1:T}^*)
    \\
&   =
   \frac{\check{\pi}_{\ell}^{(N)}(\underline{x}_{1:\ell},\underline{a}_{1:\ell-1},b_\ell)
     \gamma_{\ell:T}(x_{\ell:u}^{(b_{\ell:u})}, x_{u+1:T}^*)
     }{N^{u-\ell}}
       \prod_{k=\ell}^{{u-1}} 
    r^{(b_k,b_{k+1})}(\underline{a}_{k}\mid
    G_{k}(\underline{\vec{x}}_{k})\big)
    \!\!\! \prod_{i\neq b_{k+1}} \!\!\! M_{k+1}(x_{k+1}^{(i)}\mid
    x_k^{(a_k^{(i)})}),
\end{align*}
by Lemma \ref{lem:valid-conditional-resampling} 
\eqref{item:unconditional-vs-conditional}.
The result \eqref{item:partial-trace}
follows as we marginalise $x_{u}^{(i)}$ for $i\neq b_u$,
$\underline{a}_{u-1}$, $x_{u-1}^{(i)}$ for $i\neq b_{u-1}$, \ldots, $\underline{a}_\ell$.

For \eqref{item:partial-bridge}, define 
$\tilde{G}_k^{(\ell,u)}(x_{1:k}\mid x_u)
= G_k(x_{k-1:k}) M_{u\mid \ell}(x_u\mid
x_\ell)^{(u-\ell)^{-1}}$, and notice that 
\begin{align*}
    G_\ell(x_{\ell-1:\ell})\gamma_{\ell:T}(x_{\ell:T})
        &=\bigg(\prod_{k=\ell+1}^{u}
        \tilde{G}_{k-1}^{(\ell,u)}(x_{1:k-1}\mid x_u)
          \bar{M}_k(x_k \mid x_{k-1}, x_{u}) 
         \bigg)
         G_u(x_{u-1:u})
          \gamma_{u:T}(x_{u:T}),
\end{align*}       
where $\bar{M}_u(x_u\mid \uarg, x_u) \equiv 1$.
Adding the variables generated in lines 
\ref{line:bridge-start}--\ref{line:bridge-end}
of Algorithm \mref{alg:bridge} leads to
\begin{align}
&   \frac{\check{\pi}_{\ell}^{(N)}(\underline{x}_{1:\ell},
  \underline{a}_{1:\ell-1}, b_\ell)}{N^{u-\ell}G_{\ell}(\vec{x}_\ell^{(b_\ell)})}
    \bigg[
    \prod_{k=\ell+1}^{u-1} 
    \tilde{G}_{k-1}(\vec{\check{x}}_{k-1}^{(b_{k-1})}\mid x_u^{(b_u)})
    r^{(b_{k-1},b_k)}\big(\underline{\tilde{a}}_{k-1}\mid
    \tilde{G}_{k-1}(\underline{\vec{\check{x}}}_{k-1}\mid x_u^{(b_u)})\big)
    \label{eq:bridge-dist}
    \\
  &
    \bigg(\prod_{i=1}^N \bar{M}_k(\tilde{x}_k^{(i)}\mid
    \tilde{x}_{k-1}^{(\tilde{a}_{k-1}^{(i)})}, x_u^{(b_u)})\bigg)\bigg] 
    \tilde{G}_{u-1}(\vec{\check{x}}_{u-1}^{(b_{u-1})}\mid x_u^{(b_u)})
    G_u(x_{u-1:u}^{(b_{u-1:u})})
    \gamma_{u:T}(x_{u}^{(b_{u})}, x_{u+1:T}^*),\nonumber
\end{align}
where $\underline{\vec{\check{x}}}_\ell = \underline{x}_\ell$ and
$\vec{\check{x}}_k^{(i)} = (\vec{\check{x}}_{k-1}^{(\tilde{a}_{k-1}^{(i)})},\tilde{x}_k^{(i)})$.
Thanks to Lemma \ref{lem:valid-conditional-resampling} 
\eqref{item:unconditional-vs-conditional}
\begin{align*}
&\tilde{G}_{k-1}(\vec{\check{x}}_{k-1}^{(b_{k-1})}\mid x_u^{(b_u)})
    r^{(b_{k-1},b_k)}\big(\underline{\tilde{a}}_{k-1}\mid
    \tilde{G}_{k-1}(\underline{\vec{\check{x}}}_{k-1}\mid x_u^{(b_u)})\big)
    \\
    &= \bigg(\sum_{i=1}^N 
    \tilde{G}_{k-1}(\vec{\check{x}}_{k-1}^{(i)}\mid x_u^{(b_u)})\bigg)
    \charfun{b_{k-1} = \tilde{a}_{k-1}^{(b_k)} }
    r\big(\underline{\tilde{a}}_{k-1}\mid
    \tilde{G}_{k-1}(\underline{\vec{\check{x}}}_{k-1}\mid
    x_u^{(b_u)})\big).
\end{align*}
Because the fraction in \eqref{eq:bridge-dist} does not depend on
$b_\ell$, we may now marginalise over $b_\ell$, \ldots, $b_{u-1}$ and
add the distribution of
$\tilde{B}_{u-1}\sim\Categ(\tilde{\omega}_{u-1}^{(1:N)})$ where $\tilde{\omega}_{u-1}^{(j)} = 
\tilde{G}_{u-1}(\vec{\check{X}}_{u-1}^{(j)})
G_u(\tilde{X}_{u-1}^{(j)}, X_u^{(B_u)})$, leading into
\begin{align*}
&   \frac{\check{\pi}_{\ell}^{(N)}(\underline{x}_{1:\ell},
  \underline{a}_{1:\ell-1}, b_\ell)}{N^{u-\ell}G_{\ell}(\vec{x}_\ell^{(b_\ell)})}
    \bigg[
    \prod_{k=\ell+1}^{u-1} 
    \bigg(\sum_{i=1}^N 
    \tilde{G}_{k-1}(\vec{\check{x}}_{k-1}^{(i)}\mid x_u^{(b_u)})\bigg)
    r\big(\underline{\tilde{a}}_{k-1}\mid
    \tilde{G}_{k-1}(\underline{\vec{\check{x}}}_{k-1}\mid x_u^{(b_u)})\big)
    \\
  &\qquad\qquad
    \bigg(\prod_{i=1}^N \bar{M}_k(\tilde{x}_k^{(i)}\mid
    \tilde{x}_{k-1}^{(\tilde{a}_{k-1}^{(i)})}, x_u^{(b_u)})\bigg)\bigg] 
    \tilde{G}_{u-1}(\vec{\check{x}}_{u-1}^{(\tilde{b}_{u-1})}\mid x_u^{(b_u)})
    G_u(x_{u-1:u}^{(\tilde{b}_{u-1},b_{u})})
    \gamma_{u:T}(x_{u}^{(b_{u})}, x_{u+1:T}^*).
\end{align*}
Introducing $\tilde{b}_{\ell:u-2}$ by \textsc{AncestorTrace} leads to
addition of terms
$\charfun{\smash{\tilde{b}_{k-1}=\tilde{a}_{k-1}^{(b_k)}}}$. Then, 
calculations similar as above, but in reverse order, lead to
\eqref{eq:bridge-dist} with $b_{\ell:u-1}$ replaced with 
$\tilde{b}_{\ell:u-1}$. The result follows by marginalising
over $\tilde{X}_{\ell+1:u-1},\tilde{A}_{\ell+1:u-2}$.
\end{proof} 

\begin{proof}[Proof of Theorem \mref{thm:bbcpf-valid}] 
We start by observing that by Theorem \mref{thm:cpf-valid},
$(\underline{X}_{1:T},\underline{A}_{T-1}, \tilde{B}_T) \sim \check{\pi}_{T}^{(N)}
= \mu_{T}^{(N)}$. Then, the proof relies on an iterative application of
Lemma \ref{lem:bbcpf-recurse} \eqref{item:partial-trace} and
\eqref{item:partial-bridge}, with
$(\ell,u)=(T_{L-1},T_L),\ldots,(T_1,T_2)$, which concludes that
$(\underline{X}_1, X_{2:T}^*, \tilde{B}_{1:T}) \sim 
\mu_1^{(N)}(\underline{x}_1,\tilde{b}_1, x_{2:T}^*)/N^{T-1}$
so 
$(\tilde{X}_{1:T}^*,\tilde{B}_{1:T}) \sim \mathcal{Z}^{-1}
\gamma_{1:T}(\tilde{x}_{1:T})/N^{T} = \pi(\tilde{x}_{1:T})/N^T$.
\end{proof}

\section{Computing the conditional distributions in Assumption \mref{a:conditionals} for discretisations of linear SDEs} \label{sec:lgssm-conditionals} 

The practical application of Algorithm \mref{alg:bbcpf} requires for each block the computation of the conditional distributions $M_{u|\ell}$
and $M_{k|k-1,u}$, where $k$ is a time index, and $\ell$ and $u$ refer to the block lower and upper boundaries, respectively. 
This section discusses how these distributions may be computed when:
\begin{itemize}
  \item $M_{1:T}$ stem from a discretisation of a linear SDE 
  \item $M_{1:T}$ stem from a discretisation of a linear SDE that is conditioned on a set of noisy linear Gaussian observations.
\end{itemize}
Note that it is enough to only consider the second case, since the first one may be obtained by
omitting the conditioning on the observations (see discussion at the end of this section).

Following \citep{sarkka-solin}, the conditional means and variance (matrices) of the SDE \meqref{eq:lin-sde} are given for $t > s$ by
\begin{align}
  \E[X_t \mid X_s = x_s] &= \expm{(\mathbf{F}(t - s))} x_s \label{eq:lin-sde-cond-mean} \\
  \Var[X_t \mid X_s = x_s] &= \int_s^t \expm{(\mathbf{F}(t - \tau))} \mathbf{K} \mathbf{K}^T \expm{(\mathbf{F}(t - \tau))}^T \ud \tau, \label{eq:lin-sde-cond-var}
\end{align}
where $\expm$ denotes the matrix exponential. 
We introduce the notation 
\begin{equation}
  \label{eq:transition-notation}
  T_{s,t} := \expm{(\mathbf{F}(t - s))}, \ Q_{s, t} := \Cov[X_t \mid X_s = x_s]. 
\end{equation}
Assuming a Gaussian initial distribution, we have:
\begin{equation}
  \label{eq:lgssm}
  \begin{aligned}
    X_{t_k} \mid X_{t_{k-1}} = x_{t_{k-1}} &\sim N(T_{t_{k-1}, t_k} x_{t_{k-1}}, Q_{t_{k-1}, t_k}) \\
    X_{t_1} &\sim N(\mu_{\mathrm{init}}, \Sigma_{\mathrm{init}}),
  \end{aligned}
\end{equation}
where the time discretisation corresponds to 
\begin{equation}
  \label{eq:time-discretisation}
  0 = t_1 < t_2 < \cdots < t_T = \tau
\end{equation}
as in Section \mref{sec:fk-path-integral}, and where $\mu_{\mathrm{init}}$ and $\Sigma_{\mathrm{init}}$ are the initial mean and variance, respectively.

Suppose then that there are observations $\tilde{Y} = (\tilde{Y}_k)_{k = 1, \dots, K_y}$ observed at times $\tilde{t}_k$, $k = 1, \dots, K_y$,
where each observation time is one of the times in time discretisation \eqref{eq:time-discretisation}.
Further suppose the $\tilde{Y}_k$ are distributed as
\begin{equation}
  \label{eq:lgssm-obs-eq}
  \tilde{Y}_{k} \mid X_{\tilde{t}_k} = x_{\tilde{t}_k} \sim N(Z_{k} x_{\tilde{t}_k}, H_{k}),
\end{equation}
where $Z_k$ and $H_k$ are matrices and observation variances, respectively.
We may then define the augmented observations $Y = (Y_k)_{k = 1, \dots, T}$ with times \eqref{eq:time-discretisation}.
The random vector $Y$ is distributed like the observations $\tilde{Y}$ at their respective times, and has missing elements otherwise. 

Consider then the joint distribution of $X_{t_k}$, $X_{t_{k-1}}$ and $X_{t_u}$ for $k = \ell + 1, \dots, u - 1$, conditioned on the observations $Y_{1:T} = y_{1:T}$.
Since all variables involved are jointly Gaussian, this conditional distribution is:
\begin{equation}
  \label{eq:joint-normal}
  N\left(\begin{bmatrix}
    \mu_{k|T} \\
    \mu_{k-1|T} \\
    \mu_{u|T}
  \end{bmatrix},
  \begin{bmatrix}
    \Sigma_{k|T} & \Sigma_{k,k-1 \mid T} & \Sigma_{k, u \mid T} \\
    \Sigma_{k-1, k \mid T} & \Sigma_{k-1|T} & \Sigma_{k-1, u \mid T} \\
    \Sigma_{u, k \mid T} & \Sigma_{u, k-1 \mid T} & \Sigma_{u|T} 
  \end{bmatrix}\right),
\end{equation}
where we have used the notation
\begin{equation*}
  \begin{aligned}
    \mu_{k \mid n} &:= \E[X_{t_k} \mid Y_{1:n} = y_{1:n}], \\
    \Sigma_{k \mid n} &:= \Var[X_{t_k} \mid Y_{1:n} = y_{1:n}], \\
    \Sigma_{p, s \mid n} &:= \Cov[X_{t_p}, X_{t_s} \mid Y_{1:n} = y_{1:n}].
  \end{aligned}
\end{equation*}
Here, conditioning on a missing observation should be understood as the observation being removed from the condition.

To obtain the (cross)covariances in \eqref{eq:joint-normal}, the following backwards recursion for $s = t - 1, t - 2, \dots$ from \citep{smoothed-covariances} may be used
(with a matrix transpose applied to the result as needed):
\begin{equation}
  \label{eq:smooth-cov-recursion}
  \Sigma_{s, t \mid T} = \Sigma_{s|s}T_{t_s, t_{s + 1}}^{T}\Sigma_{s+1|s}^{-1}\Sigma_{s+1, t \mid T}.
\end{equation}
An inspection of Equations \eqref{eq:joint-normal} and \eqref{eq:smooth-cov-recursion} reveals that all the quantities required are
computed routinely by the Kalman filter and smoother \citep[cf.][]{durbin-koopman} applied to the linear Gaussian state space model 
composed of \eqref{eq:lgssm} and \eqref{eq:lgssm-obs-eq}. 
Note that the Kalman filter automatically handles any missing values in the observation sequence. 

For $k = \ell + 1$, by elementary properties of the Gaussian distribution, the distribution $M_{u|\ell}$,
that is $X_{t_u} \mid X_{t_\ell} = x_{t_{\ell}}, Y_{1:T} = y_{1:T}$, is 
\begin{equation}
  \begin{aligned}
    N\Big(&\mu_{u|T} + \Sigma_{\ell, u \mid T}^{T}\Sigma_{\ell|T}^{-1}(x_{t_{\ell}} - \mu_{\ell|T}), \\
    &\Sigma_{u|T} - \Sigma_{\ell, u \mid T}^{T}\Sigma_{\ell|T}^{-1}\Sigma_{\ell, u \mid T}\Big).
  \end{aligned}
\end{equation}

Similarly, for $k = \ell + 1, \dots, u - 1$, the distribution $M_{k \mid k - 1, u}$,
that is, $X_{t_k} \mid X_{t_{k-1}} = x_{t_{k-1}}, X_{t_u} = x_{t_u}, Y_{1:T} = y_{1:T}$, is
\begin{equation}
  \begin{aligned}
    N\Bigg(&\mu_{k|T} + \Sigma_{k, (k-1, u) \mid T}\Sigma_{(k-1, u) \mid T}^{-1}((x_{t_{k-1}} \ x_{t_u})^{T} - \mu_{(k-1, u) \mid T}), \\ 
    &\Sigma_{k|T} - \Sigma_{k, (k-1,u) \mid T}\Sigma_{(k-1, u) \mid T}^{-1}\Sigma_{k, (k-1,u) \mid T}^{T}\Bigg),
  \end{aligned}
\end{equation}
where 
\begin{equation}
  \label{eq:conditional_block_cov}
  \begin{aligned}
    \mu_{(k-1, u) \mid T} &:= (\mu_{k-1|T} \quad \mu_{u|T})^{T}, \\
    \Sigma_{k, (k-1, u) \mid T} &:= [\Sigma_{k, k-1 \mid T} \quad \Sigma_{k, u \mid T}], \\
    \Sigma_{(k-1, u) \mid T} &:= \begin{bmatrix}
      \Sigma_{k-1|T} & \Sigma_{k-1, u \mid T} \\
      \Sigma_{u, k-1 \mid T} & \Sigma_{u|T}
    \end{bmatrix}.
  \end{aligned}
\end{equation}

In the case where $M_{1:T}$ simply corresponds to a discretisation of the linear SDE \meqref{eq:lin-sde}, 
the above computations can be repeated with the conditioned means, variances and covariances replaced with their unconditional counterparts.
In practice, an easy way to compute the unconditional means and variances is to set all observations missing in the Kalman filter. 
The unconditional covariances can then be obtained from \eqref{eq:smooth-cov-recursion} as before.

\section{Models}

This section gives additional details related to the models appearing in Section \mref{sec:experiments}.

\subsection{CTCRW-P} \label{sec:ctcrwp-details}

The CTCRW-P SDE \meqref{eq:ctcrwp-sde} may be placed into the form of the linear SDE \meqref{eq:lin-sde} by setting
\[
  X_t = (V_t \ L_t)^T, \
 \mathbf{F} = \begin{bmatrix} -\beta_v & 0 \\ 1 & -\beta_x \\ \end{bmatrix}
  \  \text{and }
  \mathbf{K} = \begin{bmatrix} \sigma & 0 \\ 0 & 0 \end{bmatrix}.
\]

The expressions for $T_{s,t}$ and $Q_{s,t}$ in \eqref{eq:transition-notation} are given as follows.
A direct computation yields 
\begin{equation}
  \label{eq:ctcrwp-expm}
  \expm(Ft) = 
  \begin{bmatrix}
    \exp{(-\beta_v t)} & 0 \\ 
    \dfrac{\exp{(-\beta_x t)} - \exp{(-\beta_v t)}}{\beta_v - \beta_x}  &
    \exp{(-\beta_x t)} 
  \end{bmatrix},
\end{equation}
when $\beta_v \neq \beta_x$.
If $\beta_v = \beta_x$, the first element of the second row is replaced by $t\exp{(-\beta_v t)}$.
The transition matrix $T_{s,t}$ may be obtained from \eqref{eq:ctcrwp-expm} by substituting $t-s$ for $t$.

If $\beta_v \neq \beta_x$, the elements $q_{ij}, 1 \leq i, j \leq 2$ of $Q_{s, t}$ are given by
\begin{equation}
  \label{eq:upper-tri-ctcrwp}
  \begin{aligned}
    q_{11} &= \dfrac{\sigma^2}{2\beta_v}(1 - \exp{(-2 \beta_v(t-s))}) \\
    q_{12} &= q_{21} = \dfrac{\sigma^2}{\beta_v - \beta_x} \Bigg[ \dfrac{1}{\beta_v + \beta_x} (1 - \exp{(-(\beta_v + \beta_x)(t - s))}) - \dfrac{1}{2\beta_v}(1 - \exp{(-2\beta_v(t-s))}) \Bigg] \\
    q_{22} &= \dfrac{\sigma^2}{(\beta_v - \beta_x)^2} \Bigg[ \dfrac{1}{2\beta_x}(1 - \exp{(-2\beta_x (t - s))}) + \dfrac{1}{2\beta_v} (1 - \exp{(-2\beta_v (t - s))}) \\ 
    &\phantom{==}-\dfrac{2}{\beta_x + \beta_v} (1 - \exp{(-(\beta_x + \beta_v)(t -s))})\Bigg].
  \end{aligned}
\end{equation}
If $\beta_v = \beta_x$, the element $q_{11}$ remains as in \eqref{eq:upper-tri-ctcrwp}, but the elements $q_{12}$ and $q_{22}$ become:
\begin{equation}
  \label{eq:upper-tri-ctcrwp-eq}
  \begin{aligned}
    q_{12} &= \dfrac{\sigma^2}{4\beta_v^2}\Bigg[ 1 + \parexp{-2\beta_v (t-s)}\Big(-2\beta_v(t-s) - 1\Big)\Bigg], \\ 
    q_{22} &= \dfrac{\sigma^2}{4\beta_v^3}\Bigg[1 - \parexp{-2\beta_v(t-s)}\Big(1 + 2\beta_v(t-s)(\beta_v(t-s) + 1) \Big)\Bigg]. 
  \end{aligned}
\end{equation}

Finally, the stationary covariance matrix $S$ with elements $s_{ij}$ used in the initial distribution of the CTCRW-P model is 
obtained by taking the limit $(t - s) \to \infty$ in the previous equations:
\begin{equation}
  \label{eq:ctcrwp-statcov}
  \begin{aligned}
    s_{11} &= \dfrac{\sigma^2}{2\beta_v} \\
    s_{12} &= s_{21} = \dfrac{\sigma^2}{\beta_v - \beta_x} \Bigg[ \dfrac{1}{\beta_v + \beta_x} - \dfrac{1}{2\beta_v} \Bigg] \\
    s_{22} &= \dfrac{\sigma^2}{(\beta_v - \beta_x)^2} \Bigg[ \dfrac{1}{2\beta_x} + \dfrac{1}{2\beta_v} - \dfrac{2}{\beta_x + \beta_v} \Bigg],
  \end{aligned}
\end{equation}
when $\beta_v \neq \beta_x$. When $\beta_v = \beta_x$, the elements $s_{12}$ and $s_{22}$ are
\begin{equation}
  \label{eq:ctcrwp-statcov-eq}
  s_{12} = \dfrac{\sigma^2}{4\beta_v^2} \ \text{ and } \ s_{22} = \dfrac{\sigma^2}{4\beta_v^3}. \\
\end{equation}

\subsection{CP-RBM} \label{sec:cp-rbm-details}

The density of the reflected normal distribution $N^{(r)}(\mu, \sigma^2, a, b)$, for any point $x$ in the support $(a, b)$, is given by
\begin{equation}
  \label{eq:reflected-normal-dens}
  N^{(r)}(x; \mu, \sigma^2, a, b) = N(x; \mu, \sigma^2) + \sum_{k = 1}^{\infty} N(g_a^{(k)}(x); \mu, \sigma^2) + N(g_b^{(k)}(x); \mu, \sigma^2),
\end{equation}
where
\begin{equation}
  \label{eq:reflection-sequences}
  \begin{aligned}
    g_a^{(k)}(x) &:= (-1)^k x + ka - kb + (a + b)\charfun{\text{k odd}} \\
    g_b^{(k)}(x) &:= (-1)^k x + kb - ka + (a + b)\charfun{\text{k odd}}. \\
  \end{aligned}
\end{equation}
Equation \eqref{eq:reflected-normal-dens} may be derived by noting that the density of any point $x \in (a, b)$ is equal to the sum of normal densities at points that (eventually) reflect to $x$. 
These points consist of $x$ itself and the reflection points outside $(a, b)$ given by the sequences in \eqref{eq:reflection-sequences}.
In practice, we truncate the infinite sum in \eqref{eq:reflected-normal-dens} to the first 10 terms, which provides a reasonable approximation for the values of $\sigma$, $a$ and $b$ we use.

\subsection{CTCRW} \label{sec:ctcrw-details}

The CTCRW SDE can be placed in the form of the linear SDE \meqref{eq:lin-sde} by setting
\[
 X_t = (V_t \  L_t)^T, \
 \mathbf{F} = \begin{bmatrix} -\beta & 0 \\ 1 & 0 \\ \end{bmatrix}
  \  \text{and }
  \mathbf{K} = \begin{bmatrix} \sigma & 0 \\ 0 & 0 \end{bmatrix}.
\]

The expressions for $T_{s, t}$ and $Q_{s,t}$ in \eqref{eq:transition-notation} are then given as follows.
\begin{equation}
  T_{s, t} = \begin{bmatrix}
    \exp{(-(t - s)\beta)} & 0 \\
    \dfrac{1 - \exp{(-(t-s)\beta)}}{\beta} & 1 \\
  \end{bmatrix},
\end{equation}
and the matrix $Q_{s, t}$ has elements $q_{ij}$, $i,j = 1, 2$, such that
\begin{equation}
  \begin{aligned}
    q_{11} &= \dfrac{\sigma^2}{2\beta}\Bigg(1 - \exp{(-(t-s)2\beta)}\Bigg), \\
    q_{21} &= q_{12} = \dfrac{\sigma^2}{2\beta^2}\Bigg(1 - 2\exp{(-(t-s)\beta)} + \exp{(-(t-s)2\beta)}\Bigg), \\
    q_{22} &= \dfrac{\sigma^2}{\beta^2}\Bigg ( (t - s) - \dfrac{2}{\beta}\Big(1 - \exp{(-(t-s)\beta)}\Big) + \dfrac{1}{2\beta}\Big(1 - \exp{(-(t-s)2\beta)}\Big)\Bigg).
  \end{aligned}
\end{equation}

\section{Miscellanneous algorithms} \label{sec:misc-algs}

\subsection{Algorithm for finding mean partition order of weights} \label{sec:mean-partition-order}

The following algorithm finds a mean partition order $I$ of the input weights $w^{1:N}$. 
Note that the algorithm does not modify the weights $w^{1:N}$ and that the operation 
`break' means exiting from the current (innermost) `while' loop.
\begin{algorithm}
	\spacingset{1.0}
  \caption{\small \textsc{MeanPartitionOrder}($w^{1:N}$)}
  \label{alg:mean-partition-order}
  \begin{algorithmic}[1]
    \small
    \State Set $p = \text{\textsc{Mean}}(w^{1:N})$ (`pivot'). 
    \State Set $i_\ell = 0$ and $i_{u}$ = $N + 1$.
    \State Initialise $I$ as the index set $[N]$.
    \While{True}
      \While{$i_{\ell} < \min(i_u, N)$}
        \State Set $i_{\ell} = i_{\ell} + 1$.
        \State Break if $w^{I^{(i_{\ell})}} > p$.
      \EndWhile
      
      \While{$i_{u} > i_{\ell}$}
        \State Set $i_{u} = i_{u} - 1$.
        \State Break if $w^{I^{(i_{u})}} < p$.
      \EndWhile
      \State Break if $i_{\ell}$ equals $i_{u}$.
      \State Swap indices $i_{\ell}$ and $i_u$ of $I$.
    \EndWhile
    \State \textbf{output} $I$
  \end{algorithmic}
\end{algorithm}
\subsection{Algorithm for constructing dyadic blocking sequences} \label{sec:construct-dyadic-blockings}
$\phantom{a}$
\begin{algorithm}[h!]
	\spacingset{1.0}
  \caption{\small \textsc{DyadicCandidateBlockings}($T \in \{2, 3, \ldots\}$)} 
  \label{alg:dyadic-candidate-blockings}
  \begin{algorithmic}[1]
    \small
    \State Denote by $p^{*}$ the largest $p$ such that $2^{p} + 1 \leq T$.
    \For{$i = 1, \ldots, p^{*} + 1$}
      \State Set $\blocksize = 2^{i-1}$
      \State Set $\ell = 1; u = 0$.
      \State Set $k = 0$ (block index)
      \While{$l < T$}
        \State Set $k = k + 1$
        \State Set $u = \ell + \blocksize$
        \State Set $T_k^{(i)} = \ell; T_{k + 1}^{(i)} = \min{(u, T)}$
        \State Set $\ell = u$
      \EndWhile 
    \EndFor \\
    \Return Candidate blocking sequences $T^{(i)}_{1:L^{(i)}}$ for $i = 1, 2, \dots, p^{*} + 1$ 
  \end{algorithmic}
\end{algorithm}

\section{Derivation of $\plu_{\mathrm{G}}$} \label{sec:plu-g-derivation}

Consider the following artificial conditional particle system that approximates a continuous time conditional 
particle filter with near constant weights:
\begin{itemize}
    \item The system has $N$ particles.
    \item One of the particles corresponds to the `reference', which can not die. 
    \item At most one resampling event occurs at any time $k$, with probability $p_R^{(k)}$.
    \item If a resampling event occurs:
        \begin{itemize}
          \item a dying particle is chosen uniformly among the $N - 1$ particles
            (excluding the reference).
          \item a particle is selected for `reproduction' uniformly among the $N - 1$ particles
            (excluding the dying particle). 
        \end{itemize}
    \item If a resampling event does not occur, no particles die or reproduce.
\end{itemize}
Further suppose that the particle population is divided into two groups, `ill' and `healthy',
where the ill population is to be interpreted as the particles having reproduced from the reference
or any of its descendants. 
Denote by $H_k$ and $I_k \defeq N - H_k$ the number of healthy and number of ill (including reference)
at time $k$, respectively. 
Initially, $H_1 = N - 1$.

\begin{theorem}
  For the artificial particle system of this section, it holds for any $T \geq 1$ that 
  \begin{equation}
    \label{eq:h-expectation}
    \mathbb{E}[H_T] = (N - 1) \prod_{k=1}^{T - 1} \Bigg( 1 - \frac{p_R^{(k)}}{(N - 1)^2} \Bigg).
  \end{equation}
\end{theorem}
\begin{proof}
  For any $k \geq 2$ we have 
  \begin{equation*}
    H_k \mid H_{k-1} = 
     \begin{cases}
       H_{k-1} + 1, \ \text{ prob. } p_{\mathrm{increase}}^{(k-1)} \\ 
       H_{k-1} - 1, \ \text{ prob. } p_{\mathrm{decrease}}^{(k-1)} \\
       H_{k-1}, \ \text{ prob. } p_{\mathrm{nothing}}^{(k-1)},
     \end{cases} 
  \end{equation*}
  where
  \begin{equation*}
    \begin{aligned}
      p_{\mathrm{increase}}^{(k-1)} &= p_R^{(k - 1)} \dfrac{I_{k-1} - 1}{N - 1} \cdot \dfrac{H_{k-1}}{N - 1}\  \text{(resampling occurs, ill dies, healthy reproduces)} \\
      p_{\mathrm{decrease}}^{(k-1)} &= p_R^{(k - 1)} \dfrac{H_{k-1}}{N - 1} \cdot \dfrac{I_{k-1}}{N - 1}\  \text{(resampling occurs, healthy dies, ill reproduces)},
\\
      p_{\mathrm{nothing}}^{(k-1)} &= 1 - p_{\mathrm{increase}}^{(k-1)} - p_{\mathrm{decrease}}^{(k-1)}.
    \end{aligned}
  \end{equation*}
  Therefore, 
  \begin{equation*}
    \begin{aligned}
      \mathbb{E}[H_T \mid H_{T-1}] &= H_{T - 1} + \mathbb{E}[H_T - H_{T-1} \mid H_{T -1}] \\
      &= H_{T - 1} +  p_R^{(T - 1)} \dfrac{I_{T-1} - 1}{N - 1} \cdot \dfrac{H_{T-1}}{N - 1} - p_R^{(T - 1)} \dfrac{H_{T-1}}{N - 1} \cdot \dfrac{I_{T-1}}{N - 1} \\
      &= \Big(1 - \frac{p_R^{(T-1)}}{(N - 1)^2} \Big) H_{T - 1},
    \end{aligned}
  \end{equation*}
  and
  \begin{equation*}
    \mathbb{E}[H_T] = \mathbb{E}[\mathbb{E}[H_T \mid H_{T-1}]] = \Big(1 - \frac{p_R^{(T-1)}}{(N - 1)^2} \Big) \mathbb{E}[H_{T - 1}],
  \end{equation*}
  which yields \eqref{eq:h-expectation} by repeated application. 
\end{proof}
The direct consequence of this result is that $\plu_{\mathrm{G}}(\ell, u)$ equals $\mathbb{E}[H_u / N]$ with
$p_{R}^{(k)} = p_k N$ (defined in Section \mref{sec:blocking-seq-selection}) and $\ell$ considered as the `first' time point. 
\newpage

\section{Supplementary figures}

\begin{figure}[h!]
  \centering
  \includegraphics{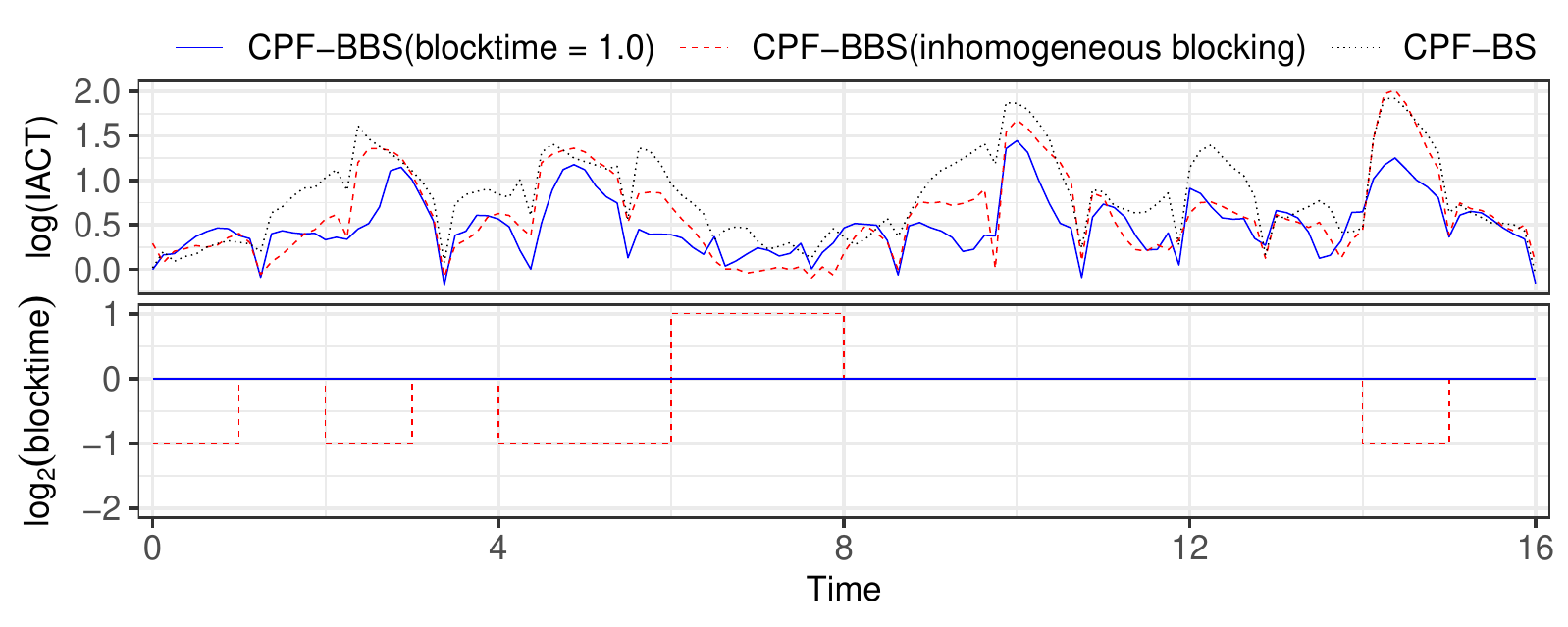}
  \caption{\small The logarithmic IACT (top) of the state $L_t^{(x)}$ in the CTCRW-T model
  with $|\Delta_k| = 0.125$ for CPF-BS and the CPF-BBS with blocktime $1.0$
  and the blocking from Algorithm \mref{alg:choose-blocking} (bottom).
  }
  \label{fig:terrain-sim-bbcpf-vs-cpfbs-iacts-supp}
\end{figure}  

\begin{figure}[h!]
  \centering
  \includegraphics{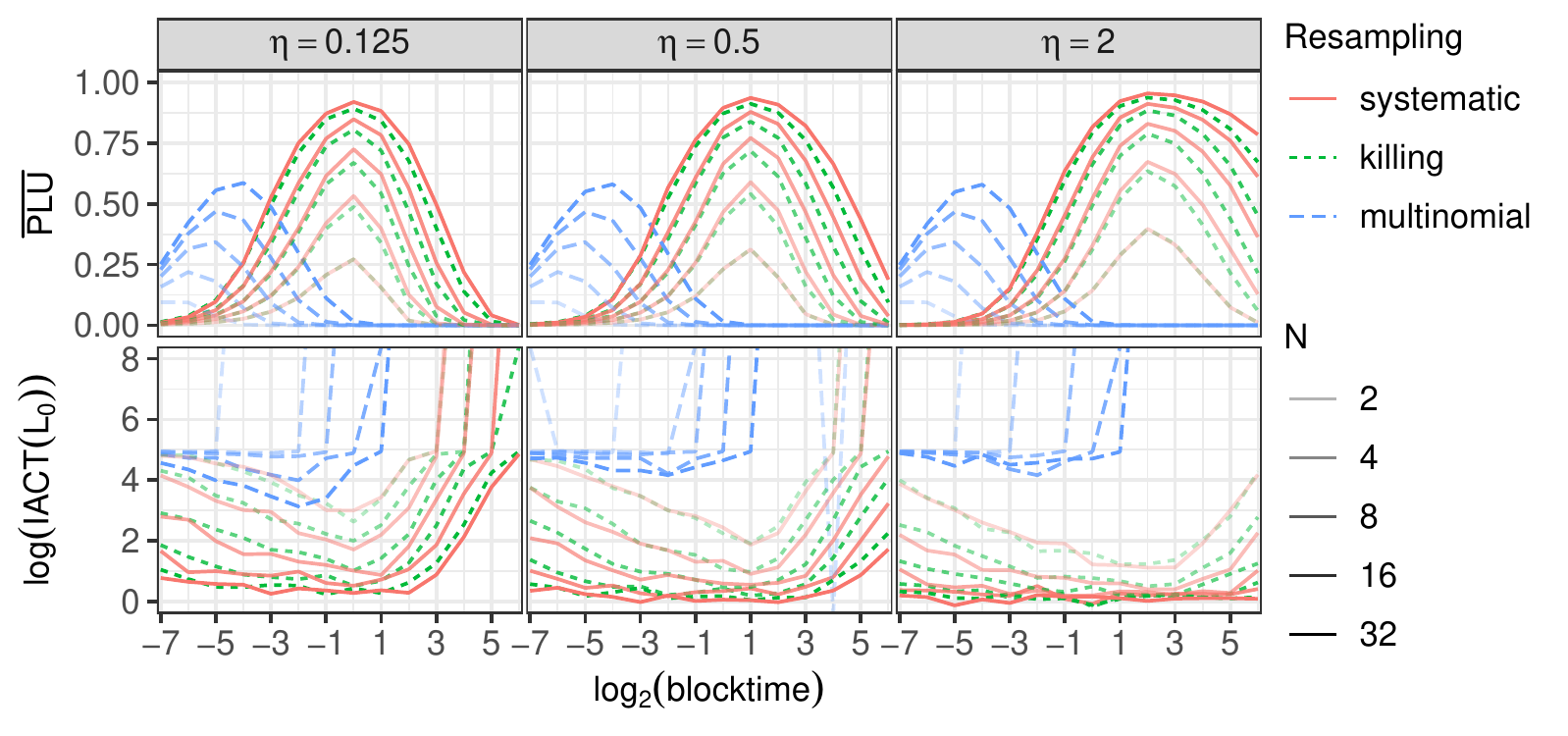}
  \caption{\small The estimated mean PLUs and the logarithm of IACT with varying $\eta$ for the location state variable at time 
  $0.0$ in the CTCRW-P model. The value of $|\Delta_k|$ was set to $2^{-7}$. 
The performance of CPF-BS is seen at the far left, with $\blocktime = 2^{-7}$.
  }
  \label{fig:ctcrwp-potential-vs-blocksize-supp}
\end{figure}

\end{document}